\documentclass[11pt]{article}

\usepackage[english]{babel}
\usepackage[T1]{fontenc}
\usepackage[margin=1in]{geometry} 
\usepackage{amsmath}
\usepackage{amsthm}
\usepackage{amssymb}
\usepackage{mathtools}
\usepackage{graphicx}
\usepackage{subcaption}
\usepackage{amssymb}
\usepackage{bm}
\usepackage{xspace}
\usepackage{xcolor}
\usepackage{float} 
\usepackage{thmtools}
\usepackage{thm-restate}
\usepackage{xfrac}
\usepackage{dsfont}

\usepackage{tikz}
\usetikzlibrary{arrows,automata}

\usetikzlibrary{arrows.meta}

\makeatletter

\pgfdeclarearrow{
  name = ipe _linear,
  defaults = {
    length = +1bp,
    width  = +.666bp,
    line width = +0pt 1,
  },
  setup code = {
    \pgfarrowssetbackend{0pt}
    \pgfarrowssetvisualbackend{
      \pgfarrowlength\advance\pgf@x by-.5\pgfarrowlinewidth}
    \pgfarrowssetlineend{\pgfarrowlength}
    \ifpgfarrowreversed
      \pgfarrowssetlineend{\pgfarrowlength\advance\pgf@x by-.5\pgfarrowlinewidth}
    \fi
    \pgfarrowssettipend{\pgfarrowlength}
    \pgfarrowshullpoint{\pgfarrowlength}{0pt}
    \pgfarrowsupperhullpoint{0pt}{.5\pgfarrowwidth}
    \pgfarrowssavethe\pgfarrowlinewidth
    \pgfarrowssavethe\pgfarrowlength
    \pgfarrowssavethe\pgfarrowwidth
  },
  drawing code = {
    \pgfsetdash{}{+0pt}
    \ifdim\pgfarrowlinewidth=\pgflinewidth\else\pgfsetlinewidth{+\pgfarrowlinewidth}\fi
    \pgfpathmoveto{\pgfqpoint{0pt}{.5\pgfarrowwidth}}
    \pgfpathlineto{\pgfqpoint{\pgfarrowlength}{0pt}}
    \pgfpathlineto{\pgfqpoint{0pt}{-.5\pgfarrowwidth}}
    \pgfusepathqstroke
  },
  parameters = {
    \the\pgfarrowlinewidth,%
    \the\pgfarrowlength,%
    \the\pgfarrowwidth,%
  },
}

\pgfdeclarearrow{
  name = ipe _pointed,
  defaults = {
    length = +1bp,
    width  = +.666bp,
    inset  = +.2bp,
    line width = +0pt 1,
  },
  setup code = {
    \pgfarrowssetbackend{0pt}
    \pgfarrowssetvisualbackend{\pgfarrowinset}
    \pgfarrowssetlineend{\pgfarrowinset}
    \ifpgfarrowreversed
      \pgfarrowssetlineend{\pgfarrowlength}
    \fi
    \pgfarrowssettipend{\pgfarrowlength}
    \pgfarrowshullpoint{\pgfarrowlength}{0pt}
    \pgfarrowsupperhullpoint{0pt}{.5\pgfarrowwidth}
    \pgfarrowshullpoint{\pgfarrowinset}{0pt}
    \pgfarrowssavethe\pgfarrowinset
    \pgfarrowssavethe\pgfarrowlinewidth
    \pgfarrowssavethe\pgfarrowlength
    \pgfarrowssavethe\pgfarrowwidth
  },
  drawing code = {
    \pgfsetdash{}{+0pt}
    \ifdim\pgfarrowlinewidth=\pgflinewidth\else\pgfsetlinewidth{+\pgfarrowlinewidth}\fi
    \pgfpathmoveto{\pgfqpoint{\pgfarrowlength}{0pt}}
    \pgfpathlineto{\pgfqpoint{0pt}{.5\pgfarrowwidth}}
    \pgfpathlineto{\pgfqpoint{\pgfarrowinset}{0pt}}
    \pgfpathlineto{\pgfqpoint{0pt}{-.5\pgfarrowwidth}}
    \pgfpathclose
    \ifpgfarrowopen
      \pgfusepathqstroke
    \else
      \ifdim\pgfarrowlinewidth>0pt\pgfusepathqfillstroke\else\pgfusepathqfill\fi
    \fi
  },
  parameters = {
    \the\pgfarrowlinewidth,%
    \the\pgfarrowlength,%
    \the\pgfarrowwidth,%
    \the\pgfarrowinset,%
    \ifpgfarrowopen o\fi%
  },
}

\newdimen\ipeminipagewidth

\tikzstyle{ipe import} = [
  x=1bp, y=1bp,
  ipe node stretch/.store in=\ipenodestretch,
  ipe stretch normal/.style={ipe node stretch=1},
  ipe stretch normal,
  ipe node/.style={
    anchor=base west, inner sep=0, outer sep=0, scale=\ipenodestretch
  },
  ipe mark scale/.store in=\ipemarkscale,
  ipe mark tiny/.style={ipe mark scale=1.1},
  ipe mark small/.style={ipe mark scale=2},
  ipe mark normal/.style={ipe mark scale=3},
  ipe mark large/.style={ipe mark scale=5},
  ipe mark normal, 
  ipe circle/.pic={
    \draw[line width=0.2*\ipemarkscale]
      (0,0) circle[radius=0.5*\ipemarkscale];
    \coordinate () at (0,0);
  },
  ipe disk/.pic={
    \fill (0,0) circle[radius=0.6*\ipemarkscale];
    \coordinate () at (0,0);
  },
  ipe fdisk/.pic={
    \filldraw[line width=0.2*\ipemarkscale]
      (0,0) circle[radius=0.5*\ipemarkscale];
    \coordinate () at (0,0);
  },
  ipe box/.pic={
    \draw[line width=0.2*\ipemarkscale, line join=miter]
      (-.5*\ipemarkscale,-.5*\ipemarkscale) rectangle
      ( .5*\ipemarkscale, .5*\ipemarkscale);
    \coordinate () at (0,0);
  },
  ipe square/.pic={
    \fill
      (-.6*\ipemarkscale,-.6*\ipemarkscale) rectangle
      ( .6*\ipemarkscale, .6*\ipemarkscale);
    \coordinate () at (0,0);
  },
  ipe fsquare/.pic={
    \filldraw[line width=0.2*\ipemarkscale, line join=miter]
      (-.5*\ipemarkscale,-.5*\ipemarkscale) rectangle
      ( .5*\ipemarkscale, .5*\ipemarkscale);
    \coordinate () at (0,0);
  },
  ipe cross/.pic={
    \draw[line width=0.2*\ipemarkscale, line cap=butt]
      (-.5*\ipemarkscale,-.5*\ipemarkscale) --
      ( .5*\ipemarkscale, .5*\ipemarkscale)
      (-.5*\ipemarkscale, .5*\ipemarkscale) --
      ( .5*\ipemarkscale,-.5*\ipemarkscale);
    \coordinate () at (0,0);
  },
  /pgf/arrow keys/.cd,
  ipe arrow normal/.style={scale=1},
  ipe arrow tiny/.style={scale=.4},
  ipe arrow small/.style={scale=.7},
  ipe arrow large/.style={scale=1.4},
  ipe arrow normal,
  /tikz/.cd,
  ipe arrows/.style={
    ipe normal/.tip={
      ipe _pointed[length=1bp, width=.666bp, inset=0bp,
                   quick, ipe arrow normal]},
    ipe pointed/.tip={
      ipe _pointed[length=1bp, width=.666bp, inset=0.2bp,
                   quick, ipe arrow normal]},
    ipe linear/.tip={
      ipe _linear[length = 1bp, width=.666bp,
                  ipe arrow normal, quick]},
    ipe fnormal/.tip={ipe normal[fill=white]},
    ipe fpointed/.tip={ipe pointed[fill=white]},
    ipe double/.tip={ipe normal[] ipe normal},
    ipe fdouble/.tip={ipe fnormal[] ipe fnormal},
    ipe arc/.tip={ipe normal},
    ipe farc/.tip={ipe fnormal},
    ipe ptarc/.tip={ipe pointed},
    ipe fptarc/.tip={ipe fpointed},
  },
  ipe arrows, 
]

\tikzset{
  rgb color/.code args={#1=#2}{%
    \definecolor{tempcolor-#1}{rgb}{#2}%
    \tikzset{#1=tempcolor-#1}%
  },
}

\makeatother

\usepackage{hyperref}
\hypersetup{
    colorlinks=true,
    linkcolor=black,
    citecolor=blue
}
\usepackage[capitalize]{cleveref}
\renewcommand*{\autoref}[1]{\cref{#1}}

\newtheorem{theorem}{Theorem}[section]
\newtheorem{lemma}[theorem]{Lemma}
\newtheorem{definition}[theorem]{Definition}
\newtheorem{remark}[theorem]{Remark}
\newtheorem{corollary}[theorem]{Corollary}
\newtheorem{fact}[theorem]{Fact}
\newtheorem{example}[theorem]{Example}
\newtheorem{policy}[theorem]{Algorithm}
\newtheorem{game}[theorem]{Definition}

\usepackage{float}
\usepackage{bookmark}
\newcommand*{\DD}{\mathcal{D}}
\newcommand{\EE}{\mathbb{E}}
\newcommand*{\EX}[2]{\underset{#1}{\mathbb{E}}\left[#2\right]}
\newcommand*{\PRO}[2]{\mathbb{P}\left[#2\right]}
\newcommand*{\PROB}[1]{\PRO{}{#1}}
\newcommand*{\PROtB}[1]{\mathbb{P}_{t_B}\left[#1\right]}

\newcommand{\C}{C}
\newcommand{\Cv}{\mathbf{C}}
\newcommand{\Cdist}{\mathcal{C}}
\newcommand{\Wv}{\mathbf{W}}
\newcommand{\X}{X}
\newcommand{\tX}{\tilde{X}}
\newcommand{\Xv}{\mathbf{X}}
\newcommand{\tXv}{\mathbf{\tilde{X}}}

\renewcommand{\SS}{\mathcal{S}}
\newcommand{\PP}{\mathcal{P}}
\renewcommand{\L}{\mathcal{L}}
\renewcommand{\O}{\mathcal{O}}

\newcommand{\1}{\mathds{1}}
\newcommand{\tails}{\mathsf{T}}
\newcommand{\heads}{\mathsf{H}}
\newcommand{\alg}{\mathsf{ALG}}
\newcommand{\opt}{\mathsf{OPT}}

\newcommand{\sspi}{SSPI}
\newcommand{\oos}{OOS\xspace}
\newcommand{\I}{\mathcal{I}}
\newcommand{\M}{\mathcal{M}}
\newcommand{\free}[2]{\mathcal{F}_{#1}(#2)}
\newcommand{\freeT}[1]{\free{\tails}{#1}}
\newcommand{\freeH}[1]{\free{\heads}{#1}}

\newcommand*{\Q}{\mathbf{Q}}
\newcommand*{\Qtill}[1]{\Q[1:#1]}
\newcommand*{\QH}{\mathcal{Q}_{\heads}}
\newcommand*{\QHtill}[1]{\QH[1:#1]}

\makeatletter
\newcommand{\printfnsymbol}[1]{%
  \textsuperscript{\@fnsymbol{#1}}%
}
\makeatother

\title{Single-Sample Prophet Inequalities Revisited}
\author{
Constantine Caramanis\\
\texttt{constantine@utexas.edu}\\
The University of Texas at Austin
\and
Matthew Faw\\
\texttt{matthewfaw@utexas.edu}\\
The University of Texas at Austin
\and
Orestis Papadigenopoulos\\
\texttt{papadig@cs.utexas.edu}\\
The University of Texas at Austin
\and
Emmanouil Pountourakis\\
\texttt{manolis@drexel.edu}\\
~~~~~~~~~~~Drexel University~~~~~~~~~~~
}
\date{\today}

\begin{document}

\maketitle

\begin{abstract}
The study of the prophet inequality problem in the limited information regime was initiated by Azar et al.\ [SODA'14]
in the pursuit of prior-independent posted-price mechanisms. As they show, $\O(1)$-competitive policies are achievable using
only a \emph{single} sample from the distribution of each agent. A notable portion of their results relies on reducing the design of single-sample prophet inequalities (\sspi{}s) to that of order-oblivious secretary (\oos{}) policies. 
The above reduction comes at the cost of not fully utilizing the available samples. However, to date, this is essentially the only method for proving \sspi{}s for many combinatorial sets (e.g., bipartite matching and various matroids). 
Very recently, Rubinstein et al.\ [ITCS'20] give a surprisingly simple algorithm
which achieves the optimal competitive ratio for the single-choice \sspi{} problem -- a result which is unobtainable going through the reduction to secretary problems.

Motivated by this discrepancy, we study the competitiveness of simple \sspi{} policies directly, without appealing to results from \oos literature. In this direction, we first develop a framework for analyzing policies against a greedy-like prophet solution. Using this framework, we obtain the first \sspi{} for general (non-bipartite)
matching environments, as well as improved competitive ratios for transversal and truncated partition matroids. Second, motivated by the observation that many \oos{} policies for matroids decompose the problem into independent rank-$1$ instances, we provide a meta-theorem which applies to any matroid satisfying this partition property. Leveraging the recent results by Rubinstein et al., we obtain improved competitive guarantees (most by a factor of $2$) for a number of matroids captured by the reduction of Azar et al. (e.g., graphic, co-graphic, and low density matroids). Finally, we discuss applications of our \sspi{}s to the design of mechanisms for multi-dimensional limited information settings with improved revenue and welfare guarantees.

\end{abstract}

\pagenumbering{arabic}

\section{Introduction}\label{sec:intro}
In optimal stopping theory, {\em prophet inequalities} serve as a fundamental framework for studying sequential decision-making problems in Bayesian environments. In the original setting, a {\em gambler} is presented with a sequence of $n$ non-negative independent random variables $\X_1, \dots, \X_n$. 
Upon observing each realization (also called ``\emph{reward}''), $X_t$, the gambler has to decide irrevocably whether to stop and collect the observed value or to forfeit it forever. The objective is to maximize the expected collected reward, compared to that of an omniscient {\em prophet}, who knows all the realizations a priori and simply stops at the maximum (thus collecting $\EE{[\max_{t} \X_t}]$ in expectation). In the above classical setting, the gambler is assumed to have distributional knowledge of the random variables, yet she has no knowledge or control over the order of arrival.

In their seminal work, Krengel, Sucheston, and Garling~\cite{KS77,KS78} prove the first optimal result for the above setting: there exists a {\em stopping rule} guaranteeing that the gambler collects at least half of the prophet's reward in expectation, assuming that $\EE[\max_t X_t] < \infty$. A simple implementation of such a stopping rule, due to Samuel-Cahn \cite{S84}, takes the form of the following {\em threshold-based} policy: Set a threshold $T$ such that $\mathbb{P}[T \geq \max_t \X_t] = \sfrac{1}{2}$ (that is, $T$ is the median of the distribution of $\max_t \X_t$) and accept the first element (if any) such that $\X_t \geq T$. Alternatively, as noted by Kleinberg and Weinberg \cite{KW12}, the same guarantee (which is provably optimal for the setting) can be achieved by setting the threshold equal to $T = \sfrac{1}{2} \cdot \EE[\max_t \X_t]$.

Due to the wide applicability and simplicity of the above model, prophet inequalities (and variations \cite{HK92}) have been an important tool and research topic in a number of fields, including theoretical computer science and algorithmic mechanism design (see \cite{CFHOV19,L17} for an overview of recent results). Specifically, prophet inequalities have been used for analyzing posted-price mechanisms for revenue maximization in sequential and multi-dimensional environments~\cite{alaei2014bayesian,chawlaMultiParam}. Naturally, this has given rise to generalizations of the model where more than one element can be chosen by the gambler, subject to combinatorial feasibility constraints (e.g., matroid, matching and knapsack constraints) \cite{KW12, R16, FSZ16, RS17, AHL12, GW19, EFGT20}.

Despite the significance of the above results, the assumption of complete distributional knowledge on the rewards is strong, and may be unrealistic in certain applications. In revenue maximization, for instance, the distribution of agents' valuations is an intrinsic function of their preferences and thus can only be learned empirically through samples. 
With a sufficient number of samples from each distribution,
one can simply use the empirically constructed counterpart to the
original prophet inequality, since it is robust to slight perturbations
of the reward distributions \cite{L17}. However, in the setting where only a limited number of samples are available,
new ideas are needed.

Motivated by these scenarios, Azar, Kleinberg, and Weinberg~\cite{AKW14,AKW18} initiated the study of the prophet inequality problem in the setting where the gambler only has access to a small number of samples from the distributions.
In this regime, they provide $\O(1)$-competitive prophet inequalities (see below for a definition) for a large number of combinatorial settings, including various types of matroids and bipartite matching. Remarkably, their results require only a \emph{single sample} from the distribution
of each reward\footnote{With the caveat that one of their results for degree-$d$ bipartite graphs requires $d^2$ samples.}.
In a very recent work, Rubinstein, Wang, and Weinberg~\cite{RWW20} show that
the original (single-choice) prophet inequality can, in fact, be solved optimally in this regime.

In light of these results, it is natural to ask if the competitive ratios obtained by \cite{AKW14,AKW18} can be improved, and if single sample prophet inequality (\sspi{}) results can be extended to different domains. Using a different approach than \cite{AKW14,AKW18}, we answer in the affirmative. We provide efficient single-sample policies with improved competitive guarantees for almost all of the combinatorial sets considered in \cite{AKW14, AKW18}. Moreover, we provide the first $\O(1)$-competitive guarantee for the single sample prophet inequality problem for non-bipartite matching -- as far as we know, the only limited-sample result available.

\subsection{Model}
We consider a ground set $E = \{e_1, e_2, \dots, e_n\}$ of $n$ elements, each associated with an unknown reward distribution $\DD_e$. Let $\DD:= \DD_{e_1} \times \DD_{e_2} \times \dots \times \DD_{e_n}$ be the product distribution of the rewards, and let $\Xv =(\X_{e_1}, \dots, \X_{e_n}) \sim \DD$ be a reward realization. The gambler (and the prophet) can choose any subset $S \subseteq E$ of elements that belongs to a given family $\I \subseteq 2^E$ of {\em feasible} sets. We assume that the gambler is given oracle access to $\I$ and a single sample $\tXv = (\tX_{e_1}, \dots, \tX_{e_n})$ from the product distribution $\DD$. In the online phase, the gambler sequentially observes the reward realization $\X_e$ of each element $e \in E$ in (potentially) adversarial order, and decides irrevocably either (i) to collect the element (if it is feasible) and obtain its reward, or (ii) to skip on the element forever. 

Let $A_t$ be the elements collected up to and including time $t$, and take
$A = A_n$ to be the final set of collected elements at the end of the arriving sequence. At any time $t$, the gambler can collect an element only if $A_{t-1} \cup \{e_t\} \in \I$ -- namely, if collecting the element does not violate feasibility. We denote by $\alg = \sum_{e \in A} X_e$ the total reward collected by the gambler and by $\opt = \max_{B\in\I} \sum_{e\in B} X_e$ the reward collected by the prophet in a given instance. 
Our goal is to design efficient {\em $\alpha$-competitive policies} for various combinatorial sets $\I$. Specifically, we seek policies that satisfy
\begin{align*}
\alpha \cdot \EX{\Xv, \tXv \sim \DD,\mathcal{R}}{\alg} \geq \EX{\Xv \sim \DD}{\opt},
\end{align*}
with the smallest possible $\alpha \geq 1$. The expectations above are taken over the randomness of the reward realizations $\Xv$, the samples $\tXv$, and the (possible) random bits $\mathcal{R}$ used by the policy.

\subsection{Main challenges and our contribution}

The state-of-the-art SSPI results for combinatorial sets come from a meta-theorem due to Azar et al.\ \cite{AKW14} that gives a reduction 
from \sspi{}s
to {\em order-oblivious secretary} (\oos) policies\footnote{We recall that the {\em secretary problem} is another important tool in optimal stopping theory. The setting is essentially the same as the prophet inequality, except for two fundamental differences: (i) The rewards of the elements are adversarially chosen and unknown to the gambler, and (ii) The elements are guaranteed to arrive in uniformly random order.}. An \oos{} policy runs in two distinct phases. In Phase $1$,
the policy chooses to sample and reject a (possibly random) number of elements. Then, in Phase $2$, the policy accepts all elements that 
meet its acceptance criteria. Crucially, the associated analysis of such a policy uses the uniformly random arrival order of the elements \emph{only} for Phase 1. Thus, the policy maintains its
competitive guarantee
\emph{even when} the elements of Phase 2 are presented
in any, potentially adversarial, order.

In \cite{AKW14}, the authors prove that {\em if there exists an $\alpha$-competitive policy for the \oos~problem on some combinatorial set, then there exists an $\alpha$-competitive policy for the corresponding \sspi~problem}. This meta-result not only implies that the \oos problem is at least as hard as the \sspi{} problem, but also yields as corollaries $\O(1)$-competitive \sspi{}s for a large class of matroids (e.g., graphic, transversal) and bipartite matching \cite{AKW18}. 

This is a notable result, and, to date, it is essentially our only tool for combinatorial \sspi{}s. We summarize their reduction in order to explain our point of departure. Let $\SS$ be an \oos policy for a given combinatorial set. We use this to 
construct the prophet policy $\PP$, which consists of two phases: offline and online. First, in the offline phase, $\PP$ simulates a random arrival order for $\SS$ by
randomly permuting the samples $\tXv$
and feeding them to the \oos{} policy until $\SS$ declares
its Phase 1 is completed. 
Then, in the online phase, if an arriving element (in adversarial order) corresponds to an element already parsed in the offline phase (Phase 1 of $\SS$), $\PP$ automatically rejects it. Otherwise, $\PP$ mirrors what the \oos policy $\SS$ does -- that is, $\PP$ accepts an element if and only if $\SS$ would accept it. Given that the sample and reward of each element are drawn independently from the same distribution, $\PP$ collects exactly the same reward in expectation as $\SS$ would collect on randomly generated rewards.

There are two important issues with above reduction. $(1)$ The actual rewards of the elements that are used in the offline phase of $\PP$ 
are never observed, since these elements are automatically skipped in the online phase. $(2)$ The samples of the elements that are not parsed in the offline phase 
are \emph{never} used. The second issue implies that reducing the \sspi{} problem to the \oos{} problem (thus, using the already known technology for the latter), as described in \cite{AKW14}, comes at the cost of sacrificing valuable information on the reward distributions, precisely because of the unobserved samples.

In order to improve on the state-of-the-art, we utilize this 
previously
sacrificed information (as described in the second point above) using two distinct approaches.
Our first approach is a natural modification of \cite{BDGIT09} from the secretary setting to the prophet inequality setting. In particular, we leverage a partition property of several matroids, which allows us to decompose the problem into simpler ones. This allows us to provide a different reduction for a wide class of matroids by using the recent result of \cite{RWW20}. For many important examples (see \cref{table:main}), this improves the competitive ratio by a factor of 2 compared to \cite{AKW14}.

Our second approach focuses on problems where the above partition property might not hold. Motivated by the recent work of \cite{RWW20}, we provide a framework for directly addressing \sspi~problems based on greedy-like algorithms. Our framework allows us to provide improved competitive guarantees for \sspi{}s for a class of combinatorial sets, including general (non-bipartite) matching and several matroids. As far as we know, this is the first such result for the general matching problem.

\subsubsection{From $\alpha$-partition property to \sspi s.}

We observe that many \oos{} policies rely on the following common principle: in Phase 1 of the \oos{} policy, the given matroid (and its ground set) is partitioned into several {\em parallel} instances of rank-1 (uniform) matroids, (potentially) as a function of the elements observed in this phase.
These instances are parallel in the sense that the union of their independent sets is always independent in the original matroid. Thus, the problem reduces to running an \oos~policy in parallel for each of these rank-1 instances\footnote{In fact, the authors in \cite{AKW14} provide a $4$-competitive \oos{} policy for rank-1 matroids. The policy samples the first half of the arriving elements and sets as a threshold the maximum observed reward. Then, for the rest of the elements, it accepts the first element of reward greater than this threshold (if one exists).}.

The process of partitioning a matroid into parallel rank-1 instances varies depending on its special structural characteristics. In fact, Babaioff et al.\ \cite{BDGIT09} refer to any matroid that can be partitioned as described above as satisfying an {\em $\alpha$-partition property} (formally defined in Section \ref{sec:reduction}). Rather than using the $4$-competitive \oos{} policy of \cite{AKW14}, this partitioning allows us to adapt the results of Rubinstein et al. \cite{RWW20}, which give a $2$-competitive policy for the \sspi~problem on rank-1 uniform matroids. This allows us to improve the competitive guarantees for a number of matroids considered in \cite{AKW14}. A technical key here is that the partition must be allowed to depend on the samples (see, for instance, the case of laminar matroids \cite{JSZ13}). We thus obtain the following meta-theorem.

\begin{restatable}{theorem}{restateReduction} \label{thm:reduction}
For any matroid $\M$ that satisfies an $\alpha$-partition property for some $\alpha \geq 1$, there exists a ${2\alpha}$-competitive policy for the corresponding \sspi~problem. Further, if the $\alpha$-partitioning can be performed in polynomial time, then the policy is also efficient.
\end{restatable}

\subsubsection{\sspi s via the ``greedy'' sample path.}
The above result is natural and straightforward, yet many interesting combinatorial sets are not currently captured by \autoref{thm:reduction}, either because they are not matroids (e.g., matching constraints), or because they might not satisfy an $\alpha$-partition property for reasonably small $\alpha$ (e.g., transversal matroids \cite{BDGIT09}). To address these problems, we develop a framework, inspired by the recent 
work of Rubinstein et al.\ \cite{RWW20}, for providing $\O(1)$-competitive policies for \sspi{} problems
directly, \emph{without} relying on the reduction from secretary problems as in \cite{AKW14}.

As already mentioned, Rubinstein et al. \cite{RWW20} provide a $2$-competitive policy for the single-choice \sspi~problem. Their policy is elegant and simple: set a threshold $T = \max_e \tX_e$ equal to the largest sample, and accept the first element whose reward exceeds $T$ (if such an element exists). Their competitive analysis relies on the following trick: for each element $e \in [n]$, instead of jointly drawing a reward and sample $\X_e, \tX_e \sim \DD_e$, we can draw two independent realizations $Y_e, Z_e \sim \DD_e$, relabel them such that $Y_e > Z_e$, and let the outcome of a fair coin flip decide whether $\X_e = Y_e$ (and $\tX_e = Z_e$), or the opposite. Based on that, the authors carefully compute the prophet's expected reward and lower bound the gambler's expected reward for any $n$ \emph{fixed} pairs of $(Y_e,Z_e)$ $\forall e \in [n]$, where the only source of randomness is that of the fair coin flips.

In our work, we extend the above idea into a framework for addressing \sspi{} problems directly, without using \oos~policies as in \cite{AKW14} or the $\alpha$-partition property for matroids. To the best of our knowledge, dedicated policies for such problems exist only for rank-1 \cite{RWW20} and uniform matroids \cite{AKW14}. The main issue here is that, since we aim to exploit the whole set of samples (as opposed to the reduction in \cite{AKW14}), we have to face the challenge of proving that the underlying {\em thresholds} of a constructed policy are {\em balanced} (in a similar sense as discussed in \cite{KW12}). Informally, these thresholds should be high enough to guarantee the quality of the obtained rewards, but small enough to accept some elements.

Specifically, in Section \ref{sec:preliminaries}, motivated by \cite{RWW20}, we construct what we call the {\em ``greedy'' sample path}: given two realizations from $\DD$, the $2n$ values are placed in non-increasing order. Then, for each element, a fair coin flip decides whether the largest realization of each distribution corresponds to the reward or sample. As we show, this construction has important properties that apply to any greedy-like offline algorithm (e.g., the standard greedy for matroids, or the greedy computation of a maximal matching in a graph). The above viewpoint substantially simplifies our setting, as it allows us to compare the expected performance of our policies against a prophet ``pointwise,'' leveraging the simplicity of fair coin flips.

The above framework allows us to compare the expected reward collected by our policies with that of a {\em greedy-like prophet}, namely, an offline (possibly near-optimal) solution computed as follows: the elements are sorted in non-increasing order of rewards, and each element in the above order is collected as long as it does not violate feasibility. For the problems we consider, the expected reward of the above greedy routine is either equal to or a $\mathcal{O}(1)$-fraction of the optimal expected reward.

At a high-level, any $\mathcal{O}(1)$-competitive policy against a greedy-like prophet must provide a sufficient condition guaranteeing that it collects (in expectation) a constant fraction of the greedy-like prophet's expected reward. This is particularly challenging when the worst-case arrival order of the elements lacks a simple characterization. In the context of \cite{RWW20}, characterizing such a sufficient condition which holds for any arrival order is relatively straightforward. In our setting, however, this becomes a central challenge, due to the nature of the combinatorial constraints. We overcome these challenges by introducing the notion of {\em supporting events}, which provide sufficient conditions for collecting heavy-weight rewards. Crucially, the probability of these events is comparable with that of the greedy-like prophet collecting the same rewards.

\subsubsection{Summary of results, properties, and connections to mechanism design.}
In \autoref{table:main}, we provide a summary of our main results. More specifically, as immediate corollaries of \autoref{thm:reduction}, we improve the competitive guarantees (by a factor of $2$ compared to \cite{AKW14}) for the cases of graphic, co-graphic, low density and column k-sparse linear matroids\footnote{Further, through \autoref{thm:reduction}, we obtain a $\frac{2e}{e-1}\approx 3.16395$-competitive
policy for $k$-uniform matroids, using Theorem $5.1$ from \cite{BDGIT09}, which may be an improvement over the {\em rehearsal algorithm} of \cite{AKW14} for small values of $k$.}.
Interestingly, the analysis for the case of graphic matroid is tight for the policy that results from our reduction. For the case of laminar matroid, we remark that the $6 \sqrt{3}$-competitive policy we provide improves on the $12 \sqrt{3}$-competitive policy presented in \cite{AKW14}, and is close to the state-of-the-art $9.6$-competitive due to \cite{AKW18}.\footnote{In \cite{AKW18}, the journal version of \cite{AKW14}, the authors provide an improved competitive guarantee of $9.6$ for the case of laminar matroid, using the \oos{} policy of \cite{ma2016simulated}.}

Using our framework, we are able to provide a $32$-competitive policy for the case of (non-bipartite) matching with edge arrivals (Section \ref{sec:matching}), a $8$-competitive policy for the transversal matroid (Section \ref{sec:transversal}), and a $8$-competitive policy for the truncated partition (i.e., two-layer laminar) matroid (Section \ref{sec:laminar}). We note that we are not aware of any results on \sspi{}s for the non-bipartite matching problem. The state-of-the-art for the bipartite case is a single-sample $256$-competitive policy \cite{AKW18}, and a $6.75$-competitive policy that applies to degree-$d$ bipartite graphs and uses $d^2$ samples \cite{AKW14}.
Surprisingly, even in the Bayesian setting, no results on prophet inequalities for the weighted matching problem on \emph{general} graphs with edge arrivals had been obtained until very recently \cite{EFGT20}.

\begin{table}%
\begin{center} 
 \begin{tabular}{||c c c c c||} 
 \hline
 Combinatorial set & Previous best & Reference & Our results & \\ [0.5ex] 
 \hline\hline
 Bipartite matching & 256 & \cite{AKW18} + \cite{FSZ18}  & 32  & Sec \ref{sec:matching}  \\ 
     & 6.75 (d-degree) & \cite{AKW14}  & 32 (any degree) & Sec \ref{sec:matching}  \\
  & ~~~ $d^2$-samples &   & 1-sample  &   \\
 General matching & -& - & 32 & Sec \ref{sec:matching} \\
 Transversal matroid & 16 & \cite{AKW14} + \cite{DP08}  & 8 & Sec \ref{sec:transversal} \\
 Laminar matroid & $9.6$ & \cite{AKW18} + \cite{ma2016simulated}  &  $6\sqrt{3} \approx 10.39$ & Thm \ref{thm:reduction} + \cite{JSZ13}\\ 
  &  &   &  8 (2-layer) & Sec \ref{sec:laminar}\\ 
     Graphic matroid & $8$ & \cite{AKW14} + \cite{KP09} & $4$ & Thm \ref{thm:reduction} + \cite{KP09} \\ 
     Co-graphic matroid & $12$ & \cite{AKW14} + \cite{Soto11} & $6$ & Thm \ref{thm:reduction} + \cite{Soto11} \\
     Low density matroid & $4\gamma(\M)$ \footnotemark  & {\cite{AKW14} + \cite{Soto11}} & $2\gamma(\M)$ & Thm \ref{thm:reduction} + \cite{Soto11} \\
     Column $k$-sparse linear matroid & {$4k$} & {\cite{AKW14} + \cite{Soto11}}  & $2k$ & Thm \ref{thm:reduction} + \cite{Soto11} \\
     [1ex] 
 \hline
\end{tabular}
\end{center}
\caption{Summary of main results}
\label{table:main}
\end{table}%
\footnotetext{The {\em density} of a matroid $\M(E, \I)$ is defined as $\gamma(\M) = \max_{S \subseteq E} \frac{|S|}{r(S)}$, where $r$ is the rank function.}

We remark that all policies that we provide have the property of being \emph{ordinal} (also known as {\em comparison-based}). This means that
they do not require accurate knowledge of the values of the random variables, but only the ability to compare any two of them. As noted in \cite{soto2018strong}, not all policies for the
secretary and prophet inequality problems work in this setting. Notably, algorithms
such as \cite{feldman2014simple}, which rely on partitioning the matroid into weight
classes, are not implementable in this model.

We also note that our competitive guarantees, similarly to \cite{RWW20}, hold against an {\em almighty adversary} -- that is, an adversary who knows \emph{all} sample and reward realizations before \emph{any} of them is revealed to the gambler. In contrast, many prophet inequalities assume either an \emph{offline} adversary, where the order is chosen \emph{before} any samples or rewards are observed, or an \emph{online weight-adaptive} adversary (as in \cite{KW12}), where the arrival order is determined sequentially, based only on the information observed by the gambler before the next element is selected.

In \autoref{sec:mechanismdesign}, we discuss the implications
of our improved \sspi{}s to mechanism design. Specifically, following the approach of Azar et al.~\cite{AKW14,AKW18}, we combine our order-oblivious {\em posted-price} algorithms with {\em lazy sample reserves}, yielding order-oblivious posted-price mechanisms with improved revenue guarantees. This comes at a cost of only a single additional sample from each reward distribution.

\subsection{Related work}
Initiated by the seminal work of Krengel, Sucheston, and Garling \cite{KS77, KS78}, the prophet inequality problem, along with its variants,
has been studied extensively in the Bayesian setting. A great deal of attention has been given to multi-choice prophet inequalities under combinatorial constraints such as
uniform \cite{hajiaghayi2007automated,alaei2014bayesian} and general matroid constraints 
\cite{chawlaMultiParam,KW12,FSZ16}.
Beyond this setting, a number of works have obtained prophet inequalities
for matchings and combinatorial auctions 
\cite{AHL12,feldman2014combinatorial,dutting2020prophet,ehsani2018prophet}. In particular, \cite{GW19} provides prophet inequalities for weighted 
bipartite matching environments under edge arrivals, and shows 
a lower bound of $2.25$ on the competitive ratio of any online
policy for this setting, showing that this problem is strictly
harder than the matroid prophet inequality.
Quite recently, \cite{EFGT20} provided the first \emph{Bayesian}
prophet inequalities for general weighted {\em(non-bipartite) matching} under edge arrivals. We emphasize, however, that nothing was known for the single-sample setting prior to the present work.
Beyond these settings, recent work has considered
the prophet inequality problem under arbitrary packing constraints
\cite{R16,RS17}.

Outside of the standard setting, variants of the problem where the arrival
order is not adversarial have been considered. Examples are the {\em free-order} model, in
which the policy can choose the arrival order \cite{yan2011mechanism,JSZ13,AKW14}, the {\em prophet-secretary} model, where the arrival order is uniformly
random \cite{ACK18,esfandiari2017prophet,correa2020prophet}, and the {\em constrained-order} model, which interpolates between the adversarial and free-order models \cite{arsenis2021constrained}. Further, variants of the problem have been studied under additional distributional assumptions, such as IID \cite{hill1982comparisons,abolhassani2017beating,correa2017posted} or dependent \cite{rinott1992optimal,immorlica2020prophet} reward distributions.

Regarding the prophet inequality problem in the limited information regime, in addition to their meta-result connecting \sspi{}s with \oos algorithms, Azar et al. \cite{AKW14, AKW18} provide a threshold-based $(1-\O(\frac{1}{\sqrt{k}}))$-algorithm for $k$-uniform matroid.
Further, Rubinstein et al. \cite{RWW20} develop a $(0.745 - \O(\epsilon))$-competitive algorithm for the IID case, using $\mathcal{O}(\frac{n}{\epsilon^6})$ samples from the distribution, improving on the results of \cite{correa2019prophet}. Additional variations of the problem under uniformly random arrival order have been recently studied in \cite{correa2020sample, CCES20}.

The motivation for studying the prophet inequality problem in many of the above settings comes largely from connections to mechanism design. While optimal single-parameter mechanisms for both welfare 
\cite{v61,g73,c97}
and revenue \cite{M81} have been well-understood for decades, there has been an extended study of simple and practical mechanisms that approximate these objectives. Prophet inequalities are known to be a powerful tool for designing simple posted-price mechanisms.
In multi-dimensional settings, much less is known about optimal mechanism design, and prophet inequalities provide one of the few
known ways to obtain approximate multi-dimensional mechanisms
\cite{chawlaMultiParam}. Azar et al.\ \cite{AKW14} show how to apply results of
\cite{chawlaMultiParam,amdw13,dry10} to obtain mechanisms in settings where the mechanism designer only has access to a single sample (or a constant number of samples) from the distribution of the agents' values.

\section{Preliminaries and Notation} \label{sec:preliminaries}
\subsection{Continuous tie-breaker}
For simplicity of exposition, we assume that any finite collection of random variables accepts almost surely a strict total ordering. While for continuous distributions the above is trivially true, for discrete distributions (or ones that contain point masses) we apply the following tie-breaking rule due to \cite{RWW20}. We can think of each distribution $\DD$ as a bivariate distribution $\DD \times U[0,1]$, where $U[0,1]$ is the uniform distribution supported in $[0,1]$. Then, for two variables $\X_1 \sim \DD_1$ and $\X_2 \sim \DD_2$ (thus, implying $(\X_e, u_e) \sim \DD_i \times U[0,1]$ for $e \in \{1,2\}$), we say that $\X_1 > \X_2$, if $\X_1 > \X_2$, or $\X_1 = \X_2$ and $u_1 > u_2$. Notice that the event $u_1 = u_2$ has zero probability. Thus, the assumption of strict total ordering for any collection of random variables does not affect our results.

\subsection{Principle of deferred decisions}
Following the paradigm of \cite{RWW20}, the reward and sample of each element $e \in E$ are generated as follows. Let $Y_e$ and $Z_e$ be two independent realizations from $\DD_e$, relabeled in a way such that $Y_e > Z_e$ (almost surely). By flipping a fair coin, we can decide whether $(\X_e, \tX_e) = (Y_e, Z_e)$ or $(\X_e, \tX_e) = (Z_e, Y_e)$ (each with probability half) -- that is, whether the largest of the two values corresponds to the reward (and the smallest to the sample), or the opposite. We refer to $Y_e$ (resp. $Z_e$) as the Y-value (resp., Z-value) of element $e \in [n]$.

\begin{fact}[\cite{RWW20}]
The above sampling procedure is equivalent to drawing $\Xv$ and $\tXv$ independently from $\DD$.
\end{fact}

Using the above fact, in \cref{sec:matching,sec:transversal,sec:laminar}, the competitive analysis of our policies is performed pointwise for \emph{any} $2n$ fixed Y/Z-values (two for each element). Thus, since our suggested policies are deterministic, the \emph{only} source of randomness is that of the fair coin flips.

\subsection{The ``greedy'' sample path}

Given a set of $2n$ values -- a $Y$ and a $Z$ value for each element $e \in E$ -- the key tool in the competitive analysis of our policies is the following construction: (1) The $2n$ values are sorted in decreasing order and the resulting sequence is relabeled as $W_1,\ldots,W_{2n}$. (2) For any {\em index} $j \in [2n]$, we denote by $e_j \in E$ the element of the ground set that corresponds to the value $W_j$, namely, $e_{j} = e$ if and only if $W_j \in \{Y_e, Z_e\}$. We refer to the $2n$-tuple $\Wv = (W_1,\ldots,W_{2n})$ as the {\em greedy (sample) path}. (3) For each $j \in [2n]$, we denote by $C_j\in\{\heads, \tails\}$ the outcome of the coin which determines whether $W_j$ is a reward (denoted by heads ``$\heads$'') or a sample (denoted by tails ``$\tails$''). We refer to the $2n$-tuple $\Cv = (\C_1, \C_2, \dots, \C_{2n}) \in \{\heads,\tails\}^{2n} $ as the {\em configuration} of the $2n$ coin flips. 

Let $\Cdist$ be the distribution of all possible configurations. We emphasize the fact that $\Cdist$ is not a product distribution, given that the pairs of coin flips corresponding to the same element are dependent. Indeed, for any element $e$ and indices $j_1 < j_2$ with $W_{j_1} = Y_e$ and $W_{j_2} = Z_e$, it either holds $(\C_{j_1}, \C_{j_2}) = (\heads, \tails)$ or $(\C_{j_1}, \C_{j_2}) = (\tails, \heads)$. Note that by slightly abusing the notation, we use $\Cdist$ to refer to also the family of all possible configurations. Finally, we notice that the probability of each feasible configuration equals $2^{-n}$.

We remark that the above construction is due to Rubinstein et al.\ \cite{RWW20} for analyzing their single-choice \sspi~policy. We extend their construction, proving additional properties that hold for any greedy-like algorithm and that apply to richer combinatorial sets (e.g., matroids and matching).

Let $\I$ be the family of feasible sets of a combinatorial problem (e.g., the independent sets of a matroid). We focus our attention on problems such that the natural greedy approach yields an optimal (or near-optimal) feasible solution (e.g. matroids or bipartite matching). The following generic algorithm produces a feasible solution to the underlying set $\I$ by parsing only the values $j \in [2n]$ of the greedy path such that $\C_j = L$, where $L \in \{\heads, \tails\}$.

\begin{definition}[Greedy$(\Wv, \Cv, L \in \{\heads, \tails\})$]
Start from the empty set $S = \emptyset$. For each index $j \in [2n]$ in increasing order such that $\C_j = L$, add $e_j$ to $S$ if and only if $S \cup \{e_j\} \in \I$.
\end{definition}

The above definition simply describes a (parameterized) version of the standard greedy algorithm that operates on the greedy path $\Wv$ and only parses entries such that $\C_j = L$ for $L \in \{\heads, \tails\}$ (i.e., either only rewards or only samples). 

\begin{definition}[Free index]\label{def:free} We say that an index $j \in [2n]$ is {\em ``free w.r.t. $\heads$''} (resp., \emph{``free w.r.t. $\tails$''}), denoted by $\freeH{j}$ (resp., $\freeT{j}$), if, given $\Wv$ and $\Cv$, the element $e_j$ can be added to the solution of the greedy algorithm without violating feasibility by the time it reaches $j$.
\end{definition}

We remark that the events $\freeH{j}$ and $\freeT{j}$ for any $j \in [2n]$ are \emph{deterministic} given the outcomes of the first $j-1$ coin flips, $\C_1,\ldots,\C_{j-1}$. Further, we emphasize that, for any $j \in [2n]$ that corresponds to a Y-value, the event $\freeH{j}$ (resp., $\freeT{j}$) {\em does not necessarily imply} that the element $e_j$ is eventually collected by the greedy algorithm, since $\C_j$ may be $\tails$ (resp., $\heads$). 

Let $\opt$ be the optimal solution for some fixed reward realization, and let $\opt'$ be the corresponding greedy solution. Note that in the case where $\I$ is the family of independent sets of a matroid, the optimality of the greedy for computing maximum weight independent sets implies that $\opt \equiv \opt'$. Additionally, we denote by $\opt_{\heads}$ (resp., $\opt_{\tails}$) the optimal solution assuming that the rewards correspond to the indices such that $\C_j = \heads$ (resp., $\C_j = \tails$). Finally, the definitions of $\opt'_{\heads}$ and $\opt'_{\tails}$ follow analogously for the case of the greedy algorithm.

We remark that, for simplicity, we slightly abuse the notation and use $\opt$ (resp., $\opt'$) to refer to both the value and the subset of the optimal (resp., greedy) solution.

\subsection{Useful properties}

We now provide useful properties of the greedy sample path. The following results are a common thread in the analysis of our policies in \cref{sec:matching,sec:transversal,sec:laminar} and hold for any fixed $\{Y_e\}_{e \in E}$ and $\{Z_e\}_{e \in E}$.

\begin{fact} \label{fact:greedy:obj}
For any fixed sample path $\Wv$, the expected reward collected by the greedy algorithm on some combinatorial set $\I$ can be expressed as:
\begin{align*}
\EX{\Cv \sim \Cdist}{\opt'} = \sum_{j \in [2n]} W_j \cdot \underset{\Cv \sim \Cdist}{\mathbb{P}}[{C_j = \heads \text{ and } \freeH{j}}].
\end{align*}
\end{fact}

For ease of notation, we drop any reference to the source of randomness which, unless otherwise noted, comes from the $n$ fair coin flips. 

The following lemma is crucial for translating the probability that a policy collects (at least) a value $W_j$ for some index $j \in [2n]$ to the probability that $W_j$ is collected by the prophet.

\begin{restatable}[Symmetry]{lemma}{restatesymmetry} \label{lem:symmetry}
For any fixed sample path $\Wv$ and index $j \in [2n]$, we have: 
$$\PRO{}{\C_j = \heads \text{ and } \freeH{j}} = \PRO{}{\C_j = \tails \text{ and } \freeT{j}}.$$ 
Further, for any $j \in [2n]$ that is a Y-value, we have:
$$\PRO{}{\C_j = \heads \text{ and } \freeT{j}} = \PRO{}{\C_j = \tails \text{ and } \freeT{j}}.$$
\end{restatable}
\begin{proof}
The first statement follows trivially, simply by exchanging the role of heads and tails. For the second statement, it suffices to notice that, for any index $j \in [2n]$ that is a Y-value, the outcome of the coin flip $\C_j$ does not affect $\freeT{j}$.
\end{proof}

A common characteristic in our proposed policies is that they never collect $Z$-values (that is, rewards that are smaller than the corresponding samples).
While this property is very helpful for the analysis of our policies, it only comes at the cost of a factor of $2$ in the competitive ratio.

\begin{restatable}[Forgetting the $Z$-values]{lemma}{restateforgettingZ} \label{lem:forgetZvals}
For any fixed greedy sample path $\Wv,$ we have: 
\begin{align*}
\sum_{\substack{j \in [2n]\\\text{Y-values}}} W_j \cdot \PRO{}{\C_j = \heads \text{ and }\freeH{j}} \geq \frac{1}{2} \sum_{j \in [2n]} W_j \cdot \PRO{}{\C_j = \heads \text{ and }\freeH{j}}.    
\end{align*}
\end{restatable}

\begin{proof}
Let $j_{(e,Y)}$ (resp. $j_{(e,Z)}$) be the index in the sample path corresponding to the Y-value (resp., Z-value) of some element $e \in E$:
\begin{align*}
   \sum_{\substack{j \in [2n]\\\text{Y-values}}} W_j \cdot \PRO{}{\C_j = \heads \text{ and }\freeH{j}} &= \frac{1}{2} \sum_{\substack{e \in E}} 2W_{j_{(e,Y)}} \cdot \PRO{}{\C_{j_{(e,Y)}} = \heads \text{ and }\freeH{j_{(e,Y)}}}  \\
   &\geq \frac{1}{2} \sum_{\substack{e \in E}} \left(W_{j_{(e,Y)}} + W_{j_{(e,Z)}}\right) \cdot \PRO{}{\C_{j_{(e,Y)}} = \heads \text{ and }\freeH{j_{(e,Y)}}}  \\
    &\geq \frac{1}{2} \sum_{\substack{j \in [2n]}}  W_j \cdot \PRO{}{\C_j = \heads \text{ and }\freeH{j}},
\end{align*}
where the first inequality follows by the fact that, for each element $e \in E$, we have $W_{j_{(e,Y)}} \geq W_{j_{(e,Z)}}$, by definition of the greedy sample path. The second inequality follows by the fact that in the greedy algorithm, the Y-value of any element $e$ has higher probability of participating in the solution comparing to the corresponding Z-value.
\end{proof}
\section{General (Non-Bipartite) Matching} \label{sec:matching}
In this section, we consider the problem of {\em maximum weighted matching} on (not necessarily bipartite) graphs under edge arrivals. Let $\mathcal{G}(V,E)$ be an undirected graph, where $V$ is the set of vertices and $E$ the set of $n$ edges. Let $\Xv \sim \DD$, where $\X_e$ is the independently realized reward of edge $e \in E$. Initially, the gambler observes a single sample $\tXv \sim \DD$, i.e., one independent sample from the distribution of each edge. In the online phase, the edges arrive in an adversarial order. After observing a reward $\X_e$ of an arriving edge, the gambler irrevocably decides whether to include this edge in the matching (if it is feasible given the already collected edges), or skip the edge forever.

Given a graph $\mathcal{G}(V,E)$ and a weight vector $\Xv$ on the edges, a {\em maximal matching} on $\mathcal{G}$ can be constructed by the following greedy routine: sort the edges in non-increasing order of weight, then greedily add each edge to the matching, if feasible, in the above order. 

Let $\opt(\mathcal{G},\Xv)$ and $\opt'(\mathcal{G}, \Xv)$ be the value of an optimal and maximal matching, respectively. The following ``folklore'' result can be found in several lecture notes on (approximation) algorithms:

\begin{restatable}{lemma}{restateMatchingMaximal} \label{lem:matching:maximal}
For any graph $\mathcal{G}(V,E)$ and weight vector $\Xv$ on the edges, we have:
$$\opt'(\mathcal{G},\Xv) \geq \frac{1}{2} \opt(\mathcal{G},\Xv).$$ 
\end{restatable}

For simplicity, we use $\opt(\mathcal{G}, \Xv)$ (resp., $\opt'(\mathcal{G},\Xv)$) to refer to both the value of the maximum (resp., maximal) matching w.r.t.\ $\Xv$ and to the matching itself. 

We consider the following policy:

\begin{policy}[\textsc{single-sample matching}] Offline, greedily compute a maximal matching $\opt'(\mathcal{G}, \tXv)$. 
For each vertex $u \in V$, define a threshold $T_u$ equal to the weight of the edge adjacent to $u$ in $\opt'(\mathcal{G}, \tXv)$, or simply $T_u = 0$ if $u$ is not matched. Online, for each arriving edge $e = \{u,v\}$, accept $e$ if and only if $\X_e > \max\{T_u, T_v\}$ and neither $u$ nor $v$ is already matched.
\end{policy}

\paragraph{Correctness and competitive analysis.} 
The correctness of the policy follows directly by the fact that the gambler never accepts an edge that is adjacent to some already matched node, and thus the set of collected edges is a feasible matching on $\mathcal{G}(V,E)$.

We now focus on the competitive analysis of our policy. Let us fix any $2n$ Y- and Z-values and let $\Wv = (W_1, \ldots, W_{2n})$ be the corresponding greedy sample path. For the fixed sample path, we denote by $\opt, \opt'$ and $\alg$ the optimal (maximum) reward, the maximal reward and our policy's reward, respectively. We also use $\opt, \opt'$ and $\alg$ to denote the set of edges in each of the three solutions. We note that these are random variables which depend on the configuration of the coin flips. 

To prove that $\alg$ collects a constant fraction of $\opt$'s reward, we show that, for every index $j \in [2n]$ which corresponds to a Y-value and such that $\opt$ collects $W_j$ with some probability, $\alg$ takes at least that value with some constant fraction of $\opt$'s acceptance probability.

Thus, consider any index $j \in [2n]$ with $e_j = \{u,v\}$.
Since $\opt'$ collects $e_j$ with value $W_j$ if and only if $\C_j = \heads$ and neither of vertices $u$ and $v$ is already matched by the time greedy reaches $j$ in $\Wv$, $\opt'$'s probability of collecting $W_j$ is given by $\PRO{}{e_j \in \opt' \text{ and } \C_j = \heads} = \PRO{}{\C_j = \heads \text{ and } \freeH{j}}$. Therefore, the expected reward of the maximal matching over the randomness of $\Cv \sim \Cdist$, is:
\begin{align*}
    \EX{}{\opt'} = \sum_{j \in [2n]} W_j \cdot \PRO{}{\C_j = \heads \text{ and } \freeH{j}}.
\end{align*}

We aim to relate this decomposition of $\opt'$ to the performance of $\alg$. To this end, for any index $j \in [2n]$ in the greedy sample path, we define the following event:

\begin{definition}[$\mathcal{G}$-Supporting event]\label{def:matchingSupportingEvent}
For any $j \in [2n]$, we say that the {\em supporting event} $\mathcal{S}_j$ occurs if the following three conditions hold simultaneously:
\begin{enumerate}
    \item $W_j$ is a Y-value satisfying $\freeT{j}$ and $\C_j = \heads$.
    \item For the smallest index $\ell_1 > j$ in the greedy sample path, where $e_{\ell_1}$ is adjacent, parallel or identical to $e_j$, such that $\freeT{\ell_1}$, it holds that $\C_{\ell_1} = \tails$.
    \item If $e_{\ell_1}$ is not parallel or identical to $e_j$, for the smallest index $\ell_2 > \ell_1$ in the path such that $e_{\ell_2}$ is adjacent to $e_j$ and satisfies $\freeT{\ell_2}$, if such an index exists, it holds that $\C_{\ell_2} = \tails$.
\end{enumerate}
\end{definition}

As we will show, for any index $j \in [2n]$ with $e_j = \{u,v\}$, the event $\mathcal{S}_j$ is a sufficient condition guaranteeing that, in the online phase, at least one of vertices $u$ and $v$ is matched with an edge of reward at least $W_j$. The requirement that $j$ is a Y-value 
is enforced automatically by the policy, as it can be easily shown that Z-values are never accepted by our policy.
We first provide an example which demonstrates the conditions enforced by the supporting event:

\begin{example}
Consider any index $j \in [2n]$ with $e_j = \{u,v\}$ such that $W_j$ is a Y-value satisfying $\freeT{j}$ and $\C_j = \heads$. Let $\ell_1$ be the smallest index after $j$ in the sample path such that $\freeT{\ell_1}$ and $e_{\ell_1}$ share at least one common vertex with $e_j$. In Figure \ref{fig:matching:supporting1}, we see an example where index $\ell_1$ corresponds to the same edge as $j$, while in Figure \ref{fig:matching:supporting2} the edges $e_{\ell_1}$ and $e_{j}$ are parallel. Finally, Figure \ref{fig:matching:supporting3} describes a situation where the index $\ell_1$ corresponds to an edge adjacent to vertex $u$ (but not $v$), and index $\ell_2$ corresponds to an edge adjacent to vertex $v$.
\end{example}

\begin{figure}[!htb]
\minipage{0.32\textwidth}
\resizebox{1\linewidth}{!}{%
\tikzstyle{ipe stylesheet} = [
  ipe import,
  even odd rule,
  line join=round,
  line cap=butt,
  ipe pen normal/.style={line width=0.4},
  ipe pen heavier/.style={line width=0.8},
  ipe pen fat/.style={line width=1.2},
  ipe pen ultrafat/.style={line width=2},
  ipe pen normal,
  ipe mark normal/.style={ipe mark scale=3},
  ipe mark large/.style={ipe mark scale=5},
  ipe mark small/.style={ipe mark scale=2},
  ipe mark tiny/.style={ipe mark scale=1.1},
  ipe mark normal,
  /pgf/arrow keys/.cd,
  ipe arrow normal/.style={scale=7},
  ipe arrow large/.style={scale=10},
  ipe arrow small/.style={scale=5},
  ipe arrow tiny/.style={scale=3},
  ipe arrow normal,
  /tikz/.cd,
  ipe arrows, 
  <->/.tip = ipe normal,
  ipe dash normal/.style={dash pattern=},
  ipe dash dashed/.style={dash pattern=on 4bp off 4bp},
  ipe dash dotted/.style={dash pattern=on 1bp off 3bp},
  ipe dash dash dotted/.style={dash pattern=on 4bp off 2bp on 1bp off 2bp},
  ipe dash dash dot dotted/.style={dash pattern=on 4bp off 2bp on 1bp off 2bp on 1bp off 2bp},
  ipe dash normal,
  ipe node/.append style={font=\normalsize},
  ipe stretch normal/.style={ipe node stretch=1},
  ipe stretch normal,
  ipe opacity 10/.style={opacity=0.1},
  ipe opacity 30/.style={opacity=0.3},
  ipe opacity 50/.style={opacity=0.5},
  ipe opacity 75/.style={opacity=0.75},
  ipe opacity opaque/.style={opacity=1},
  ipe opacity opaque,
]
\definecolor{red}{rgb}{1,0,0}
\definecolor{green}{rgb}{0,1,0}
\definecolor{blue}{rgb}{0,0,1}
\definecolor{yellow}{rgb}{1,1,0}
\definecolor{orange}{rgb}{1,0.647,0}
\definecolor{gold}{rgb}{1,0.843,0}
\definecolor{purple}{rgb}{0.627,0.125,0.941}
\definecolor{gray}{rgb}{0.745,0.745,0.745}
\definecolor{brown}{rgb}{0.647,0.165,0.165}
\definecolor{navy}{rgb}{0,0,0.502}
\definecolor{pink}{rgb}{1,0.753,0.796}
\definecolor{seagreen}{rgb}{0.18,0.545,0.341}
\definecolor{turquoise}{rgb}{0.251,0.878,0.816}
\definecolor{violet}{rgb}{0.933,0.51,0.933}
\definecolor{darkblue}{rgb}{0,0,0.545}
\definecolor{darkcyan}{rgb}{0,0.545,0.545}
\definecolor{darkgray}{rgb}{0.663,0.663,0.663}
\definecolor{darkgreen}{rgb}{0,0.392,0}
\definecolor{darkmagenta}{rgb}{0.545,0,0.545}
\definecolor{darkorange}{rgb}{1,0.549,0}
\definecolor{darkred}{rgb}{0.545,0,0}
\definecolor{lightblue}{rgb}{0.678,0.847,0.902}
\definecolor{lightcyan}{rgb}{0.878,1,1}
\definecolor{lightgray}{rgb}{0.827,0.827,0.827}
\definecolor{lightgreen}{rgb}{0.565,0.933,0.565}
\definecolor{lightyellow}{rgb}{1,1,0.878}
\definecolor{black}{rgb}{0,0,0}
\definecolor{white}{rgb}{1,1,1}
\begin{tikzpicture}[ipe stylesheet]
  \draw[ipe pen ultrafat]
    (128, 640)
     -- (192, 704);
  \draw[ipe pen ultrafat]
    (192, 704)
     -- (192, 736);
  \draw[ipe pen ultrafat]
    (192, 704)
     -- (224, 704);
  \draw[ipe pen ultrafat]
    (128, 640)
     -- (128, 608);
  \draw[ipe pen ultrafat, ipe dash dotted]
    (192, 736)
     -- (192, 768);
  \draw[ipe pen ultrafat, ipe dash dotted]
    (224, 704)
     -- (256, 704);
  \draw[ipe pen ultrafat, ipe dash dotted]
    (128, 608)
     -- (128, 576)
     -- (128, 576);
  \draw[ipe pen ultrafat]
    (192, 704)
     arc[start angle=-8.1301, end angle=98.1301, x radius=56.5685, y radius=-56.5685];
  \node[ipe node, font=\huge]
     at (112, 640) {$u$};
  \node[ipe node, font=\huge]
     at (176, 704) {$v$};
  \node[ipe node, font=\huge]
     at (98.811, 685.714) {$e_j \equiv e_{\ell_1}$};
  \pic
     at (192, 704) {ipe disk};
  \pic
     at (128, 640) {ipe disk};
  \pic[fill=white]
     at (192, 704) {ipe fdisk};
  \pic[fill=white]
     at (128, 640) {ipe fdisk};
  \node[ipe node, font=\LARGE, text=red]
     at (208, 720) {$T_v = W_{\ell_1}$};
  \node[ipe node, font=\LARGE, text=red]
     at (144, 624) {$T_u = W_{\ell_1}$};
\end{tikzpicture}}
  \caption{$e_{\ell_1}$ is identical to $e_j$}\label{fig:matching:supporting1}
\endminipage\hfill
\minipage{0.32\textwidth}
\resizebox{1\linewidth}{!}{%
\tikzstyle{ipe stylesheet} = [
  ipe import,
  even odd rule,
  line join=round,
  line cap=butt,
  ipe pen normal/.style={line width=0.4},
  ipe pen heavier/.style={line width=0.8},
  ipe pen fat/.style={line width=1.2},
  ipe pen ultrafat/.style={line width=2},
  ipe pen normal,
  ipe mark normal/.style={ipe mark scale=3},
  ipe mark large/.style={ipe mark scale=5},
  ipe mark small/.style={ipe mark scale=2},
  ipe mark tiny/.style={ipe mark scale=1.1},
  ipe mark normal,
  /pgf/arrow keys/.cd,
  ipe arrow normal/.style={scale=7},
  ipe arrow large/.style={scale=10},
  ipe arrow small/.style={scale=5},
  ipe arrow tiny/.style={scale=3},
  ipe arrow normal,
  /tikz/.cd,
  ipe arrows, 
  <->/.tip = ipe normal,
  ipe dash normal/.style={dash pattern=},
  ipe dash dashed/.style={dash pattern=on 4bp off 4bp},
  ipe dash dotted/.style={dash pattern=on 1bp off 3bp},
  ipe dash dash dotted/.style={dash pattern=on 4bp off 2bp on 1bp off 2bp},
  ipe dash dash dot dotted/.style={dash pattern=on 4bp off 2bp on 1bp off 2bp on 1bp off 2bp},
  ipe dash normal,
  ipe node/.append style={font=\normalsize},
  ipe stretch normal/.style={ipe node stretch=1},
  ipe stretch normal,
  ipe opacity 10/.style={opacity=0.1},
  ipe opacity 30/.style={opacity=0.3},
  ipe opacity 50/.style={opacity=0.5},
  ipe opacity 75/.style={opacity=0.75},
  ipe opacity opaque/.style={opacity=1},
  ipe opacity opaque,
]
\definecolor{red}{rgb}{1,0,0}
\definecolor{green}{rgb}{0,1,0}
\definecolor{blue}{rgb}{0,0,1}
\definecolor{yellow}{rgb}{1,1,0}
\definecolor{orange}{rgb}{1,0.647,0}
\definecolor{gold}{rgb}{1,0.843,0}
\definecolor{purple}{rgb}{0.627,0.125,0.941}
\definecolor{gray}{rgb}{0.745,0.745,0.745}
\definecolor{brown}{rgb}{0.647,0.165,0.165}
\definecolor{navy}{rgb}{0,0,0.502}
\definecolor{pink}{rgb}{1,0.753,0.796}
\definecolor{seagreen}{rgb}{0.18,0.545,0.341}
\definecolor{turquoise}{rgb}{0.251,0.878,0.816}
\definecolor{violet}{rgb}{0.933,0.51,0.933}
\definecolor{darkblue}{rgb}{0,0,0.545}
\definecolor{darkcyan}{rgb}{0,0.545,0.545}
\definecolor{darkgray}{rgb}{0.663,0.663,0.663}
\definecolor{darkgreen}{rgb}{0,0.392,0}
\definecolor{darkmagenta}{rgb}{0.545,0,0.545}
\definecolor{darkorange}{rgb}{1,0.549,0}
\definecolor{darkred}{rgb}{0.545,0,0}
\definecolor{lightblue}{rgb}{0.678,0.847,0.902}
\definecolor{lightcyan}{rgb}{0.878,1,1}
\definecolor{lightgray}{rgb}{0.827,0.827,0.827}
\definecolor{lightgreen}{rgb}{0.565,0.933,0.565}
\definecolor{lightyellow}{rgb}{1,1,0.878}
\definecolor{black}{rgb}{0,0,0}
\definecolor{white}{rgb}{1,1,1}
\begin{tikzpicture}[ipe stylesheet]
  \draw[ipe pen ultrafat]
    (128, 640)
     -- (192, 704);
  \draw[ipe pen ultrafat]
    (192, 704)
     -- (192, 736);
  \draw[ipe pen ultrafat]
    (192, 704)
     -- (224, 704);
  \draw[ipe pen ultrafat]
    (128, 640)
     -- (128, 608);
  \draw[ipe pen ultrafat, ipe dash dotted]
    (192, 736)
     -- (192, 768);
  \draw[ipe pen ultrafat, ipe dash dotted]
    (224, 704)
     -- (256, 704);
  \draw[ipe pen ultrafat, ipe dash dotted]
    (128, 608)
     -- (128, 576)
     -- (128, 576);
  \draw[ipe pen ultrafat]
    (192, 704)
     arc[start angle=-8.1301, end angle=98.1301, x radius=56.5685, y radius=-56.5685];
  \node[ipe node, font=\huge]
     at (112, 640) {$u$};
  \node[ipe node, font=\huge]
     at (176, 704) {$v$};
  \node[ipe node, font=\huge]
     at (137.353, 679.568) {$e_j$};
  \pic
     at (192, 704) {ipe disk};
  \pic
     at (128, 640) {ipe disk};
  \pic[fill=white]
     at (192, 704) {ipe fdisk};
  \pic[fill=white]
     at (128, 640) {ipe fdisk};
  \node[ipe node, font=\huge]
     at (185.769, 650.777) {$e_{\ell_1}$};
  \node[ipe node, font=\LARGE, text=red]
     at (208, 720) {$T_v = W_{\ell_1}$};
  \node[ipe node, font=\LARGE, text=red]
     at (144, 624) {$T_u = W_{\ell_1}$};
\end{tikzpicture}}
  \caption{$e_{\ell_1}$ is parallel to $e_j$}\label{fig:matching:supporting2}
\endminipage\hfill
\minipage{0.32\textwidth}%
\resizebox{1\linewidth}{!}{%
\tikzstyle{ipe stylesheet} = [
  ipe import,
  even odd rule,
  line join=round,
  line cap=butt,
  ipe pen normal/.style={line width=0.4},
  ipe pen heavier/.style={line width=0.8},
  ipe pen fat/.style={line width=1.2},
  ipe pen ultrafat/.style={line width=2},
  ipe pen normal,
  ipe mark normal/.style={ipe mark scale=3},
  ipe mark large/.style={ipe mark scale=5},
  ipe mark small/.style={ipe mark scale=2},
  ipe mark tiny/.style={ipe mark scale=1.1},
  ipe mark normal,
  /pgf/arrow keys/.cd,
  ipe arrow normal/.style={scale=7},
  ipe arrow large/.style={scale=10},
  ipe arrow small/.style={scale=5},
  ipe arrow tiny/.style={scale=3},
  ipe arrow normal,
  /tikz/.cd,
  ipe arrows, 
  <->/.tip = ipe normal,
  ipe dash normal/.style={dash pattern=},
  ipe dash dashed/.style={dash pattern=on 4bp off 4bp},
  ipe dash dotted/.style={dash pattern=on 1bp off 3bp},
  ipe dash dash dotted/.style={dash pattern=on 4bp off 2bp on 1bp off 2bp},
  ipe dash dash dot dotted/.style={dash pattern=on 4bp off 2bp on 1bp off 2bp on 1bp off 2bp},
  ipe dash normal,
  ipe node/.append style={font=\normalsize},
  ipe stretch normal/.style={ipe node stretch=1},
  ipe stretch normal,
  ipe opacity 10/.style={opacity=0.1},
  ipe opacity 30/.style={opacity=0.3},
  ipe opacity 50/.style={opacity=0.5},
  ipe opacity 75/.style={opacity=0.75},
  ipe opacity opaque/.style={opacity=1},
  ipe opacity opaque,
]
\definecolor{red}{rgb}{1,0,0}
\definecolor{green}{rgb}{0,1,0}
\definecolor{blue}{rgb}{0,0,1}
\definecolor{yellow}{rgb}{1,1,0}
\definecolor{orange}{rgb}{1,0.647,0}
\definecolor{gold}{rgb}{1,0.843,0}
\definecolor{purple}{rgb}{0.627,0.125,0.941}
\definecolor{gray}{rgb}{0.745,0.745,0.745}
\definecolor{brown}{rgb}{0.647,0.165,0.165}
\definecolor{navy}{rgb}{0,0,0.502}
\definecolor{pink}{rgb}{1,0.753,0.796}
\definecolor{seagreen}{rgb}{0.18,0.545,0.341}
\definecolor{turquoise}{rgb}{0.251,0.878,0.816}
\definecolor{violet}{rgb}{0.933,0.51,0.933}
\definecolor{darkblue}{rgb}{0,0,0.545}
\definecolor{darkcyan}{rgb}{0,0.545,0.545}
\definecolor{darkgray}{rgb}{0.663,0.663,0.663}
\definecolor{darkgreen}{rgb}{0,0.392,0}
\definecolor{darkmagenta}{rgb}{0.545,0,0.545}
\definecolor{darkorange}{rgb}{1,0.549,0}
\definecolor{darkred}{rgb}{0.545,0,0}
\definecolor{lightblue}{rgb}{0.678,0.847,0.902}
\definecolor{lightcyan}{rgb}{0.878,1,1}
\definecolor{lightgray}{rgb}{0.827,0.827,0.827}
\definecolor{lightgreen}{rgb}{0.565,0.933,0.565}
\definecolor{lightyellow}{rgb}{1,1,0.878}
\definecolor{black}{rgb}{0,0,0}
\definecolor{white}{rgb}{1,1,1}
\begin{tikzpicture}[ipe stylesheet]
  \draw[ipe pen ultrafat]
    (128, 640)
     -- (192, 704);
  \draw[ipe pen ultrafat]
    (192, 704)
     -- (192, 736);
  \draw[ipe pen ultrafat]
    (128, 640)
     -- (128, 608);
  \draw[ipe pen ultrafat, ipe dash dotted]
    (192, 736)
     -- (192, 768);
  \draw[ipe pen ultrafat, ipe dash dotted]
    (128, 608)
     -- (128, 576)
     -- (128, 576);
  \node[ipe node, font=\huge]
     at (112, 640) {$u$};
  \node[ipe node, font=\huge]
     at (176, 704) {$v$};
  \node[ipe node, font=\huge]
     at (137.353, 679.568) {$e_j$};
  \pic
     at (192, 704) {ipe disk};
  \pic
     at (128, 640) {ipe disk};
  \pic[fill=white]
     at (192, 704) {ipe fdisk};
  \pic[fill=white]
     at (128, 640) {ipe fdisk};
  \node[ipe node, font=\huge]
     at (169.769, 618.777) {$e_{\ell_1}$};
  \draw[ipe pen ultrafat]
    (128, 640)
     -- (192, 592);
  \draw[ipe pen ultrafat]
    (192, 704)
     -- (256, 704);
  \pic[fill=white]
     at (128, 640) {ipe fdisk};
  \pic[fill=white]
     at (192, 704) {ipe fdisk};
  \node[ipe node, font=\huge]
     at (217.769, 682.777) {$e_{\ell_2}$};
  \node[ipe node, font=\LARGE, text=red]
     at (208, 720) {$T_v = W_{\ell_2}$};
  \node[ipe node, font=\LARGE, text=red]
     at (160, 640) {$T_u = W_{\ell_1}$};
\end{tikzpicture}}
\caption{$e_{\ell_1}$ and $e_{\ell_2}$ are adjacent to $e_j$}\label{fig:matching:supporting3}
\endminipage
\end{figure}

We now outline three crucial properties of $\mathcal{S}_j$ that naturally lead to the proof of our main result. We provide their proofs at the end of this section. The first lemma says that for any vertex $u \in V$, there exists at most one index $j$ with $e_j$ adjacent to $u$ such that the supporting event is true. 

\begin{lemma}\label{lem:matching:uniqueness}
For any $u \in V$, it holds: $\sum_{j \in [2n]} \1\{\mathcal{S}_j\text{ and }e_j \text{ is adjacent to }u\} \leq 1$.
\end{lemma}

Let $Q(u)$ be the reward of the edge adjacent to $u$ in the final solution of our policy. As we show in the next lemma, if an index $j \in [2n]$ with corresponding edge $e_j = \{u,v\}$ satisfies the supporting event $\mathcal{S}_j$, then at least one of $u$ and $v$ is matched in $\alg$ with an edge of reward at least $W_j$.

\begin{lemma} \label{lem:matching:sufficiency}
For any $j \in [2n]$ and corresponding edge $e_j = \{u,v\}$ such that the supporting event $\mathcal{S}_j$ is true, we have that $
    Q(u) + Q(v) \geq W_j$.
\end{lemma}

Finally, for any index $j$ that corresponds to a Y-value, we can associate the probability of $\mathcal{S}_j$ with the probability that the value $W_j$ appears in the maximal solution $\opt'$.

\begin{lemma} \label{lem:matching:probability}
For any $j \in [2n]$ that corresponds to a Y-value, we have: $\PRO{}{\mathcal{S}_j} \geq \frac{1}{4} \PRO{}{\C_j = \heads \text{ and } \freeH{j}}$.
\end{lemma}

By combining the above properties of the supporting event, we are now ready to prove our main result.

\begin{theorem}\label{thm:matchingCompetitive}
\textsc{single-sample matching} is $32$-competitive for (non-bipartite) matching.
\end{theorem}
\begin{proof}
As discussed above, by \autoref{fact:greedy:obj}, we can express the expected reward collected by the maximal matching as
\begin{align*}
    \EX{}{\opt'} = \sum_{j \in [2n]} W_j \cdot \PRO{}{\C_j = \heads \text{ and } \freeH{j}}.
\end{align*}

Recall that $Q(u)$ is the reward of the edge adjacent to a vertex $u \in V$ in the solution returned by our algorithm, where $Q(u) = 0$ if no edge in $\alg$ is adjacent to $u$. We can express the reward collected by our algorithm as follows:

\begin{align*}
    \alg = \frac{1}{2} \sum_{u \in V} Q(u) \geq \frac{1}{2} \sum_{u \in V} Q(u) \sum_{j \in [2n]} \1\{\mathcal{S}_j\text{ and }e_j \text{ is adjacent to }u\},
\end{align*}
where the inequality follows by \autoref{lem:matching:uniqueness}.

By exchanging the order of summations in the RHS of the above equality, we have that:
\begin{align*}
    \alg \geq \frac{1}{2} \sum_{\substack{j \in [2n]\\\text{with } e_j = \{u,v\}}} \left(Q(u) + Q(v)\right)\1\{\mathcal{S}_j\text{ and }e_j=\{u,v\}\}. 
\end{align*}
By taking the expectation over the randomness of the configurations, we get:
\begin{align*}
    \EX{}{\alg} \geq \frac{1}{2} \sum_{\substack{j \in [2n]\\\text{with } e_j = \{u,v\}}} \EX{}{\left(Q(u) + Q(v)\right)\1\{\mathcal{S}_j\text{ and }e_j=\{u,v\}\}} \geq \frac{1}{2} \sum_{\substack{j \in [2n] \\ \text{Y-values}}} W_j\cdot \PRO{}{\mathcal{S}_j},
\end{align*}
where the second inequality follows by \autoref{lem:matching:sufficiency}. Note that we can restrict ourselves to indices that are Y-values, since no Z-value can satisfy the supporting event. 

For any index $j$ that is a Y-value, by \autoref{lem:matching:probability}, we have that $\PRO{}{\mathcal{S}_j} \geq \frac{1}{4} \PRO{}{\C_j = \heads \text{ and } \freeH{j}}$. Thus, we get that:

\begin{align*}
\EX{}{\alg} \geq \frac{1}{2} \sum_{\substack{j \in [2n] \\ \text{Y-values}}} W_j\cdot \PRO{}{\mathcal{S}_j} \geq \frac{1}{8} \sum_{\substack{j \in [2n] \\ \text{Y-values}}} W_j\cdot \PRO{}{\C_j = \heads \text{ and } \freeH{j}}.
\end{align*}

By \autoref{lem:forgetZvals} the contribution of the Y-values to the expected reward of $\opt'$, as computed by the greedy, is at least half of the total expected reward. Therefore, we can conclude that:
\begin{align*}
\EX{}{\alg} \geq \frac{1}{16} \sum_{j \in [2n]} W_j\cdot \PRO{}{\C_j = \heads \text{ and } \freeH{j}} = \frac{1}{16} \EX{}{\opt'}.
\end{align*}

The proof follows by \autoref{lem:matching:maximal}, since the expected reward of a maximal matching is at least half of the expected reward of the optimal matching (collected by the prophet), thus, 
\begin{align*}
    \EX{}{\alg} &\geq \frac{1}{16} \EX{}{\opt'} \geq \frac{1}{32} \EX{}{\opt}. 
    \qedhere
\end{align*}
\end{proof}

We complete this section by proving the three lemmas related to the supporting event.

\begin{proof}[Proof of \autoref{lem:matching:uniqueness}]
We fix any vertex $u \in V$ and let $E_u \subseteq E$ be the edges of the graph that are adjacent to $u$. We would like to show that $\sum_{j \in [2n]} \1\{\mathcal{S}_j\text{ and }e_j \in E_u\} \leq 1$. Suppose that there exists a configuration of coin flips such that this is not true, and, for this configuration, let $j_1$ be the smallest and $j_2$ be the second smallest index in the sample path such that $e_{j_1}, e_{j_2} \in E_u$ and the events $\mathcal{S}_{j_1}$ and $\mathcal{S}_{j_2}$ are true. By definition of $\mathcal{S}_{j_1}$, it has to be that $j_1$ is a Y-value with $\C_{j_1} = \heads$ and $\freeT{j_1}$. Further, for the smallest index $\ell_{1} > j_1$ such that $\freeT{\ell_1}$ and $e_{\ell_1}$ share at least one common vertex with $e_{j_1}$, it has to be that $\C_{\ell_1} = \tails$. Notice that since $j_2$ satisfies $\mathcal{S}_{j_2}$, we necessarily have that $j_1 < \ell_1 < j_2$. Indeed, in the opposite case, $j_2$ would be the smallest integer after $j_1$ satisfying $\freeT{j_2}$ and such that $e_{j_2}$ shares at least one common vertex with $e_{j_1}$. However, by definition of $\mathcal{S}_{j_1}$, this in turn would imply that $\C_{j_2} = \tails$. 
We distinguish between two cases: 

(i) Suppose $e_{\ell_1}$ is adjacent to vertex $u$. By definition of $\mathcal{S}_{j_1}$, it has to be that $\C_{\ell_1} = \tails$ and $\freeT{\ell_1}$ and, thus, $e_{\ell_1}$ participates in $\opt'_T$, that is in the maximal matching w.r.t. samples. This implies that no other index $i > \ell_1$ with $e_i \in E_u$ can satisfy $\freeT{i}$. Hence, we get a contradiction, since $j_2$ cannot satisfy (the first property of) $\mathcal{S}_{j_2}$. 

(ii) Suppose $e_{\ell_1}$ is not adjacent to vertex $u$ (this is possible, since $e_{\ell_1}$ can be adjacent to the other endpoint of $e_{j_1}$). Let $\ell_2$ be the smallest index $\ell_2 > \ell_1$ such that $\freeT{\ell_2}$ and $e_{\ell_2} \in E_u$. Clearly, if such an index does not exist, then $j_2$ cannot satisfy $\mathcal{S}_{j_2}$ (as $j_2 > \ell_1$ is an index after $j_1$ such that $e_{j_2} \in E_u$ and $\freeT{j_2}$, by definition of $\mathcal{S}_{j_2}$). Notice, further that it must be that $j_2 > \ell_2$, since, in the opposite case, that would imply that $j_2 \equiv \ell_2$ and, thus, $\C_{j_2} = \tails$ which contradicts $\mathcal{S}_{j_2}$.
Again, by definition of $\mathcal{S}_{j_1}$, it has to be that $\C_{\ell_2} = \tails$ and, thus, $e_{\ell_2}$ is collected by the greedy algorithm in the maximal matching $\opt'_T$. Thus, there cannot exist any index $i > \ell_2$ with $e_i \in E_u$ such that $\freeT{i}$. This leads to a contradiction to the fact that $\mathcal{S}_{j_2}$ is true, since $j_2 > \ell_2$, by definition of $\ell_2$.
\end{proof}

\begin{proof}[Proof of \autoref{lem:matching:sufficiency}] We need to show that, for any $j \in [2n]$ such that $W_j$ is a Y-value with corresponding edge $e_j = \{u,v\}$, the event $\mathcal{S}_j$ implies that at least one of $u,v \in V$ is matched in $\alg$ with an edge of reward at least $W_j$. Let $E_u , E_v \subseteq E$ be the edges adjacent to vertices $u$ and $v$, respectively. The first step it to show that, assuming $\mathcal{S}_j$ is true, the thresholds $T_u$ and $T_v$ as defined by our policy in the offline phase are both smaller than $W_j$. That would imply that the edge $e_j$ with reward $W_j$ would be accepted by the policy, unless one of $u,v$ is already matched.

Let $\ell_{1} > j$, be the smallest index after $j$ such that $\freeT{\ell_1}$ and $e_{\ell_1}$ is adjacent to vertex $u$, $v$ or both, as described in the definition of $\mathcal{S}_j$. By assumption, $j$ is a Y-value, so the index $\ell_1$ is guaranteed to exist (as in the limiting case it can be the corresponding Z-value of $e_j$). 
We distinguish between two cases: 

(i) In the case where $e_{\ell_1}$ is parallel to $e_j$ or $e_{\ell_1} \equiv e_{j}$ (that is, indices $j$ and $\ell_1$ correspond to the same edge), then $e_{\ell_1}$ participates in the maximal matching $\opt'_T$, given that $\C_{\ell_1} = \tails$, by definition of $\mathcal{S}_j$. Thus, the thresholds of $u$ and $v$ in that case satisfy $T_u = T_v = W_{\ell_1} < W_j$. See Figures \ref{fig:matching:supporting1} and \ref{fig:matching:supporting2} for an illustration.

(ii) In the case where $e_{\ell_1}$ is not parallel or identical to $e_j$, let us assume w.l.o.g. that $e_{\ell_1} \in E_u$. Since $\C_{\ell_1} = \tails$ (by assumption of $\mathcal{S}_j$), it must be that $e_{\ell_1} \in \opt'_T$ and, thus, $T_u = W_{\ell_1}$. Consider now the index $\ell_2$, that is the smallest index after $\ell_1$ such that $\freeT{\ell_2}$ and $e_{\ell_2} \in E_v$. Notice that since $e_{\ell_1}$ is adjacent to $u$ and induces a threshold $W_{\ell_1}$, it cannot be that $e_{\ell_2}$ is parallel or identical to $e_j$. Note, further, that if $\ell_2$ does not exist, then vertex $v$ is not matched in $\opt'_T$ and, thus, $T_v = 0 < W_j$. In the case where $\ell_2$ exists, then by $\mathcal{S}_j$ it has to be that $\C_{\ell_2} = \tails$ and, thus, $e_{\ell_2} \in \opt'_T$. This implies that for the threshold of vertex $v$, we have that $T_v = W_{\ell_2} < W_j$, as illustrated in Figure \ref{fig:matching:supporting3}.

By the above analysis, we can see that, when the supporting event $\mathcal{S}_j$ is true, then it must be that $W_j > \max\{T_u, T_v\}$. Therefore, assuming that $e_j$ is the first arriving edge in the online phase and given that $\C_j = \heads$ (again by definition of $\mathcal{S}_j$), it is definitely collected by our algorithm, which would in turn imply that $Q(u) + Q(v) \geq W_j$.

Consider now the time where $e_j$ arrives in the online phase (not necessarily first). By the event $\mathcal{S}_j$, it is the case that $\C_{j} = \heads$ and, thus, the value of $e_j$ in the online phase is $W_j$. If $e_j$ cannot be collected, it is because at least one of $u$ and $v$ are already matched using a different edge. To conclude the proof, it suffices to show that the reward of this edge is never less than $W_j$.
Consider an index $i$, such that the edge $e_i$ of reward $W_i$ is matched to $u$, $v$ or both (before the arrival of $e_j$). Clearly, in the case where $i < j$, it follows that $Q(u) + Q(v) \geq W_i > W_j$. It suffices to show that it cannot be that $i > j$. We first note that in order for the algorithm to collect $e_i$ with reward $W_i$, apart from $\C_i = \heads$, it has to hold that $\freeT{i}$. Indeed, in the opposite case, that would imply that (at least) one of the endpoints of $e_i$ are already matched in $\opt'_T$ using some edge of sample value greater than $W_i$ and, thus, $e_i$ cannot be collected as it cannot exceed both thresholds. Now, since $i>j$ and assuming $\freeT{i}$, then, by definition of $\mathcal{S}_j$, it has to be that either $i\equiv \ell_1$ or $i \equiv \ell_2$. However, this leads to a contradiction, since $\mathcal{S}_j$ implies that $\C_{\ell_1} = \C_{\ell_2} = \tails$. 
\end{proof}

\begin{proof}[Proof of \autoref{lem:matching:probability}]
Let us fix any index $j \in [2n]$ that corresponds to the Y-value of edge $e_j = \{u,v\}$. Recall that we denote by $E_u$ the edges of $E$ that are adjacent to vertex $u$. We would like to show that $\mathbb{P}[\mathcal{S}_j] \geq \frac{1}{4} \mathbb{P}[\C_j = \heads \text{ and }\freeH{j}]$. By \autoref{lem:symmetry}, and since $j$ is a Y-value, we immediately get that 
$\mathbb{P}[\C_j = \heads \text{ and }\freeH{j}] = \mathbb{P}[\C_j = \heads \text{ and }\freeT{j}]$, so equivalently we need to show that $\mathbb{P}[\mathcal{S}_j] \geq \frac{1}{4} \mathbb{P}[\C_j = \heads \text{ and }\freeT{j}]$.

We provide a proof of the above inequality via counting arguments on the set of configurations. Let $\Cdist^{0}\subseteq \Cdist$ be the subset of possible configurations that satisfy $\C_j = \heads$ and $\freeT{j}$. Starting from any configuration $\Cv \in \Cdist^0$, our goal is to transform the configuration into one that satisfies $\mathcal{S}_j$. This can be achieved in two steps: First, we transform $\Cv$ to verify that index $\ell_1$ (as defined in the definition of $\mathcal{S}_j$), satisfies $\C_{\ell_1} = \tails$. Then, we further transform the configuration in a way such that the index $\ell_2$ (if it exists) also satisfies $\C_{\ell_2} = \tails$. Using the fact that all configurations in $\Cdist$ are equiprobable (each having probability $2^{-n}$), we show that at least a $\frac{1}{4}$ fraction of the configurations in $\Cdist^{0}$ also satisfies $\mathcal{S}_j$. The above idea is depicted in Figures \ref{fig:matching:prob1}, \ref{fig:matching:prob2} and \ref{fig:matching:prob3}.

Let $\Cdist^1 \subseteq \Cdist^0 \subseteq \Cdist$ be the family of configurations such that $\C_j = \heads$, $\freeT{j}$ and, in addition, for the smallest index $\ell_1 > j$ with $\freeT{\ell_1}$ and such that $e_{\ell_1}$ shares at least one vertex with $e_j$, it holds $\C_{\ell_1} = \tails$. Consider any configuration $\Cv \in \Cdist^0 \setminus \Cdist^1$ and let $\ell_1$ be the smallest index after $j$ as described above. Since $\Cv \notin \Cdist^1$, it has to be that $\C_{\ell_1} = \heads$. We construct a configuration $\Cv'$ that is identical to $\Cv$, except for the coin corresponding to $e_{\ell_1}$ (including both Y- and Z-values), that is now flipped such that $\C'_{\ell_1} = \tails$. Notice that if $\ell_1$ is a Y-value, it immediately follows that $\Cv' \in \Cdist^1$ (since $\ell_1$ still satisfies $\freeT{\ell_1}$ after flipping $\C_{\ell_1}$). In addition, one can see that
$\ell_1$ {\em cannot be} a Z-value.
Indeed, in the case where $\ell_1$ is a Z-value satisfying $\freeT{\ell_1}$ (by assumption), then the corresponding Y-value, let $\ell_1'$, should also satisfy $\freeT{\ell'_1}$, by construction of the greedy algorithm that computes the maximal matching $\opt'_T$. However, in configuration $\Cv$, since $\C_{\ell_1}= \heads$, that would imply that the corresponding Y-value satisfies $\C_{\ell'_1}= \tails$. This fact, in combination with $\freeT{\ell'_1}$, means that $e_{\ell_1}$ with value $W_{\ell'_1}$ participates in $\opt'_T$, which in turn contradicts $\freeT{\ell_1}$. Therefore, for the constructed configuration it holds $\Cv' \in \Cdist^1$.

By the above analysis, we can conclude that for any configuration $\Cv \in \Cdist^0 \setminus \Cdist^1$ there exists a configuration $\Cv'$ such that $\Cv' \in \Cdist^1$. Further, we note that for any $\Cv \in \Cdist^0 \setminus \Cdist^1$, the produced configuration $\Cv' \in \Cdist^1$ is unique, since $\Cv'$ cannot be produced as described above by starting from a configuration different than $\Cv$. Therefore, given the fact that all configurations are equiprobable, we directly get that $\PRO{}{\Cv \in \Cdist^1} \geq \frac{1}{2} \PRO{}{\Cv \in \Cdist^0}$.

\begin{figure}[!htb] 
\minipage{0.32\textwidth}
\resizebox{1.0\linewidth}{!}{%
\begin{tabular}{|c c c c c c| }
& $j$ & $\dots$ & $\ell_1$ & $\dots$ & $\ell_2$ \\ 
& $\freeT{j}$ & $\dots$ & $\freeT{\ell_1}$ & $\dots$ & $\freeT{\ell_2}$ \\  
& $\heads$ & $\dots$ & $\heads$ & $\dots$ & $\heads$ 
\end{tabular}
}
\caption{$\Cv \in \Cdist^0$.} \label{fig:matching:prob1}
\endminipage\hfill
\minipage{0.32\textwidth}
\resizebox{1.0\linewidth}{!}{%
\begin{tabular}{c c c c c c }
$j$ & $\dots$ & $\ell_1$ & $\dots$ & $\ell_2$ \\ 
$\freeT{j}$ & $\dots$ & $\freeT{\ell_1}$ & $\dots$ & $\freeT{\ell_2}$ \\  
$\heads$ & $\dots$ & $\color{red} \bf T$ & $\dots$ & $\heads$
\end{tabular}
}
\caption{$\Cv' \in \Cdist^1$}\label{fig:matching:prob2}
\endminipage\hfill
\minipage{0.32\textwidth}%
\resizebox{1.0\linewidth}{!}{%
\begin{tabular}{|c c c c c c| }
$j$ & $\dots$ & $\ell_1$ & $\dots$ & $\ell_2$ \\ 
$\freeT{j}$ & $\dots$ & $\freeT{\ell_1}$ & $\dots$ & $\freeT{\ell_2}$\\  
$\heads$ & $\dots$ & $\bf T$ & $\dots$ & $\color{red} \bf T$
\end{tabular}
}
\caption{$\Cv'' \in \Cdist^2$}\label{fig:matching:prob3}
\endminipage
\end{figure}

In a similar way as above, we construct a set of configurations $\Cdist^2 \subseteq \Cdist^1$ such that the third property of $\mathcal{S}_j$ is satisfied. Consider any configuration $\Cv' \in \Cdist^1$. If the edge $e_{\ell_1}$ (with $\ell_1$ as described above) is identical or parallel to $e_j$ (as in Figures \ref{fig:matching:supporting1} and \ref{fig:matching:supporting2}), we simply add $\Cv'$ to $\Cdist^2$ (as in that case the third property of $\mathcal{S}_j$ holds trivially). Suppose that $e_{\ell_1}$ shares exactly one vertex with $e_j$, let w.l.o.g. $e_{\ell_1} \in E_u$ (as in Figure \ref{fig:matching:supporting3}). In that case, let $\ell_2$ be the first index after $\ell_1$ such that $\freeT{\ell_2}$ and $e_{\ell_2} \in E_v$. In case $\ell_2$ does not exist or if $\C'_{\ell_2} = \tails$, we simply add $\Cv'$ to $\Cdist^2$. In the case where $\C'_{\ell_2} = \heads$, we can construct a new configuration $\Cv''$ that is identical to $\Cv'$, except for the the value of $\C''_{\ell_2}$, where we now set to $\C''_{\ell_2} = \tails$, and we add $\Cv''$ to $\Cdist^2$. Note, further, that by the exact same arguments as in the case of $\ell_1$ above, we can show that the index $\ell_2$ satisfies $\freeT{\ell_2}$ in $\Cv''$.

Again, any configuration $\Cv'' \in \Cdist^2$ can be produced in a unique way by some $ \Cv' \in \Cdist^1$ as described above. Since all the configurations are equiprobable and by a simple counting argument, we can see that $\PRO{}{\Cv \in \Cdist^2} \geq \frac{1}{2} \PRO{}{\Cv \in \Cdist^1} \geq \frac{1}{4} \PRO{}{\Cv \in \Cdist^0}$. The lemma follows by the fact that the event $\mathcal{S}_j$ is satisfied for any configuration in $\Cdist^2$ and, thus, $\PRO{}{\mathcal{S}_j} \geq \frac{1}{4} \PRO{}{\Cv \in \Cdist^0} = \frac{1}{4} \PRO{}{\C_j = \heads \text{ and }\freeT{j}}$.
\end{proof}

\section{Transversal Matroid} \label{sec:transversal}
We now consider the case of {\em transversal matroids}. Let $\mathcal{G}(L \cup R, A)$ be an undirected bipartite graph, where $L$ is the set of $|L| = n$ {\em left} vertices (L-nodes), $R$ is the set of {\em right} vertices (R-nodes), and $A$ is the set of edges. In the transversal matroid over the ground set of L-nodes, any subset $S \subseteq L$ is independent (meaning feasible) if the vertices of $S$ can be perfectly matched with (a subset) of the R-nodes. To avoid confusion, we emphasize that this setting is very different than the general matching of Section \ref{sec:matching}, as now the $n$ elements correspond to the L-\emph{nodes} in the above bipartite graph. As before, let $\Xv \sim \DD$ be the vector of rewards and $\tXv \sim \DD$ be the vector of samples observed offline by the gambler. As usual, in the online phase, the gambler observes the reward $\X_l$ of each arriving L-node $l$ and, then, irrevocably decides whether to collect $l$ and match it with some adjacent right vertex (if this is feasible), or to skip on the vertex.

Motivated by the techniques of Dimitrov and Plaxton \cite{DP08}, we define the notion of an {\em ordered-maximal (bipartite) matching},
which is constructed in the following manner.
First, select an arbitrary ordering on the R-nodes, $r_1, \ldots, r_{|R|}$. Then, for every $l \in L$ in non-increasing order of
weight, match $l$ with the {\em smallest} (in the chosen
ordering on R-nodes) adjacent R-node that is not already matched. 
If all adjacent R-nodes of $l$ are matched, $l$ remains unmatched. 

As before, we denote by $\opt(\mathcal{G},\Xv)$ and $\opt'(\mathcal{G}, \Xv)$ both the value and the subset of matched L-nodes in the optimal and ordered-maximal case, respectively, for a bipartite graph $\mathcal{G}$ and weight vector $\Xv$.

\begin{restatable}[\cite{DP08}]{lemma}{restateTransversalOrdered} \label{lem:transversal:ordered}
For any bipartite graph $\mathcal{G}(L \cup R, A)$ and weight vector $\Xv$ on the L-nodes, we have 
$$\opt'(\mathcal{G},\Xv) \geq \frac{1}{2} \opt(\mathcal{G},\Xv).$$ 
\end{restatable}

We consider the following policy:

\begin{policy}[\textsc{single-sample transversal}]
Offline, choose an arbitrary ordering on R and, given this, compute an ordered-maximal matching $\opt'(G, \tXv)$. For each $r \in R$, set a threshold $T_r$ equal to the weight of the L-node adjacent to $r$ in $\opt'(\mathcal{G}, \tXv)$, or $T_r = 0$, if $r$ remains unmatched. Online, for each arriving $l \in L$, find the smallest $r \in R$ which is a neighbor of $l$ and $\X_l > \max\{T_r, \tX_l\}$. If $r$ exists and is not already matched, accept $l$ and \emph{match it to $r$}; otherwise, skip on $l$.
\end{policy}

\paragraph{Correctness and competitive analysis.} 
The correctness of the policy follows by the fact that the underlying perfect matching between $L$ and $R$-nodes is constructed simultaneously with the collection of L-nodes in the online phase. Thus, the collected set is trivially independent for the transversal matroid.

We now analyze the competitive guarantee of our algorithm. As in Section \ref{sec:matching}, the competitive analysis is performed pointwise for any fixed Y- and Z-values. For any $2n$ fixed such values, let $\Wv = (W_1, \ldots, W_{2n})$ be the corresponding greedy sample path. In the following, we use $\opt, \opt'$ and $\alg$ to refer both to the value and the set itself of the optimal, the ordered-maximal, and our algorithm, respectively. Finally, we refer to the value and solution of the ordered-maximal matching w.r.t. samples as $\opt'_T$.

Fix any arbitrary ordering $r_1, \dots, r_{|R|}$ on the $R$-nodes. In the case of transversal matroid, the event $\freeH{j}$ (resp., $\freeT{j}$) for some index $j$ denotes the fact that, by the time $e_j$ is parsed by the greedy procedure which computes an ordered-maximal matching w.r.t.\ the rewards (resp., samples), the R-node 
which matches to $l_j$ is unmatched.

The following definition is an adaptation of a similar definition in \cite{DP08} to the context of \sspi{}s:
\begin{definition}[Candidate node]\label{def:candidateNodeTransversal}
For any index $j \in [2n]$ such that $\freeT{j}$, we denote $f_T(j)\in R$ as the index of the R-node which would be matched with $l_j$ in $\opt'_T$ (determined by the fixed ordering on R-nodes) if $\C_j$ were $\tails$.
\end{definition}

For any index $j \in [2n]$ in the greedy path and vertex $r \in R$, we define the following supporting event, which plays a similar
role to \autoref{def:matchingSupportingEvent} in the matching case:

\begin{definition}[$r$-Supporting event]
For any index $j \in [2n]$ and $r \in R$, we say that the {\em supporting event} $\mathcal{S}_{j,r}$ occurs if the following conditions hold simultaneously:
\begin{enumerate}
    \item $W_j$ is a Y-value satisfying $\freeT{j}$ and $\C_j = \heads$.
    \item For vertex $r$, we have $r = f_T(j)$.
    \item Let $\ell > j$ be the smallest index in the greedy sample path such that $\ell$ is a $Y$-value, $\freeT{\ell}$ is true, and $r = f_T(\ell)$. If $\ell$ exists, then $\C_{\ell} = \tails$.
\end{enumerate}
\end{definition}
 
We now provide useful properties of $\mathcal{S}_{j,r}$. Their role is analogous to the ones in Section \ref{sec:matching} and their proofs follow a similar spirit (yet they are not identical).

The following lemma states that, for any vertex $r \in R$, there exists at most one index $j \in [2n]$ such that the supporting event $\mathcal{S}_{j,r}$ holds.
\begin{restatable}{lemma}{restateTransversalUniqueness} \label{lem:transversal:uniqueness}
For any vertex $r \in R$, it holds: $\sum_{\substack{j \in [2n]\\ \text{Y-values}}} \1\{\mathcal{S}_{j,r}\} \leq 1$.
\end{restatable}

Let $Q(r)$ be the reward of the L-node adjacent to $r \in R$ in $\alg$. As we show, if $j \in [2n]$ satisfies $\mathcal{S}_{j,r}$ for some $r \in R$, then the latter is matched in $\alg$ with an L-node of reward at least $W_j$.

\begin{restatable}{lemma}{restateTransversalSufficiency} \label{lem:transversal:sufficiency}
For any $j \in [2n]$, we have that $\sum_{r \in R} Q(r) \cdot \1\{\mathcal{S}_{j,r}\} \geq W_j\cdot \1\{ \exists r \in R \text{ such that }\mathcal{S}_{j,r}\}$.
\end{restatable}

Finally, we can associate the probability that the sufficient event for an index $j$ occurs with the probability that the value $W_j$ appears in the ordered-maximal solution $\opt'$.

\begin{restatable}{lemma}{restateTransversalProbability} \label{lem:transversal:probability}
For any Y-value $j \in [2n]$, we have: $\PRO{}{\exists r \in R \text{ such that }\mathcal{S}_{j,r}} \geq \frac{1}{2} \PRO{}{\C_j = \heads \text{ and } \freeH{j}}$.
\end{restatable}

By combining the above results, we are able to provide our main result in
a similar manner to the proof of \autoref{thm:matchingCompetitive}: 

\begin{restatable}{theorem}{restateTransversalCompetitiveGuarantee}
\textsc{single-sample transversal} is $8$-competitive for the transversal matroid.
\end{restatable}

\begin{proof}
By \autoref{fact:greedy:obj}, we can express the expected reward collected by an ordered-maximal matching as 
$$
\EX{}{\opt'} = \sum_{j \in [2n]} W_j \cdot \PRO{}{\C_j = \heads \text{ and } \freeH{j}}.
$$

Recall that $Q(r)$ is the reward of the L-node adjacent to $r \in R$ in the solution of our algorithm, where $Q(r) = 0$ if $r$ is unmatched. We can express the reward collected by our algorithm as follows:

\begin{align*}
    \alg = \sum_{r \in R} Q(r) &\geq \sum_{\substack{j \in [2n]\\ \text{Y-values}}} \sum_{r \in R} Q(r) \cdot \1\{\mathcal{S}_{j,r}\} 
    \geq \sum_{\substack{j \in [2n]\\ \text{Y-values}}} W_j \cdot \1\{ \exists r \in R \text{ such that }\mathcal{S}_{j,r}\},
\end{align*}
where the first inequality follows by \autoref{lem:transversal:uniqueness}, and the second by \autoref{lem:transversal:sufficiency}. 

By taking the expectation over the random coin flips in the above expression, we get:
\begin{align*}
    \EX{}{\alg} \geq \sum_{\substack{j \in [2n]\\ \text{Y-values}}} W_j \cdot \PRO{}{\exists r \in R \text{ such that }\mathcal{S}_{j,r}} \geq \frac{1}{2} \sum_{\substack{j \in [2n]\\ \text{Y-values}}} W_j \cdot \PRO{}{\C_j = \heads\text{ and }\freeH{j}},
\end{align*}
where the second inequality follows by \autoref{lem:transversal:probability}.

By using \autoref{lem:forgetZvals}, we have that:
\begin{align*}
\EX{}{\alg} & \geq \frac{1}{2} \sum_{\substack{j \in [2n]\\ \text{Y-values}}} W_j \cdot \PRO{}{\C_j = \heads\text{ and }\freeH{j}} \geq \frac{1}{4} \sum_{\substack{j \in [2n]}} W_j \cdot \PRO{}{\C_j = \heads\text{ and }\freeH{j}}.
\end{align*}

Finally, by \autoref{lem:transversal:ordered} we have that $\EX{}{\opt'} = \sum_{\substack{j \in [2n]}} W_j \cdot \PRO{}{\C_j = \heads\text{ and }\freeH{j}} \geq \frac{1}{2}\EX{}{\opt}$, thus giving $\EX{}{\alg} \geq \frac{1}{8} \EX{}{\opt}$.
\end{proof}
\section{Truncated Partition Matroid} \label{sec:laminar}
We consider a ground set $E$ of $n$ elements, and a family of disjoint subsets $\L = \{ E_1,E_2, \ldots, E_l\}$ such that $\bigcup_{i \in [l]} E_i = E$. Each subset $E_i \in \L$ is associated with a {\em capacity} $r_{E_i} \in \mathbb{N}_{\geq 1}$, while their union (that is, the whole ground set) is associated with a capacity $r_E \in \mathbb{N}_{\geq 1}$. Given the above setting, a {\em truncated partition matroid} $\M=(E,\I)$ is defined such that any subset $S \subseteq E$ is independent if and only if $|S\cap E_i|\leq r_{E_i}$ for every $E_i \in \L$ and, in addition, $|S| \leq r_E$. We remark that the above definition is a special case of a {\em laminar matroid} over a (laminar) family $\L \subset 2^E$ of subsets of $E$. Recall that $\L$ is called laminar if, for any two sets $S,T \in \L$, we either have $S \subseteq T$, $T \subseteq S$, or $S \cap T = \emptyset$ (that is, $\L$ does not contain any {\em crossing} subsets).

Let $\opt(\Xv)$ (resp., $\opt(\tXv)$) be the optimal solution w.r.t. the rewards (resp., samples). We remark that in the case of matroids, the standard greedy algorithm returns an exact solution.
We denote by $(\Xv_{-e}, \X'_e )$ the reward vector $\Xv$ after replacing its $e$-th coordinate with $\X'_e$. We consider the following algorithm for the \sspi~problem on laminar matroids:

\begin{policy}[\textsc{single-sample laminar}] \label{alg:laminar}
Offline, greedily compute the optimal solution $\opt(\tXv)$ w.r.t. the samples. 
Online, at each arriving element $e$, compute $\opt(\tXv_{-e}, \X_e)$ -- that is, recompute $\opt(\tXv)$ after replacing the sample $\tX_e$ with the observed reward $\X_e$. If collecting $e$ does not violate independence, accept the element if and only if $\opt(\tXv_{-e}, \X_e) > \opt(\tXv)$; otherwise, skip on $e$.
\end{policy}

We remark that the above algorithm is a natural adaptation of the $9.6$-competitive secretary policy of Ma et al. \cite{ma2016simulated} for laminar matroid to the \sspi~setting. For the special case of truncated partition matroid, we are able to provide an $8$-approximation through a different analysis based on our greedy sample path framework.

\paragraph{Correctness and competitive analysis.} 
The correctness of the algorithm follows trivially since independence of the collected set is enforced at every iteration.

We now focus our attention to the competitive analysis of our policy. Let us fix any $2n$ Y- and Z-values and let $\Wv = (W_1, \ldots, W_{2n})$ be the greedy sample path. Again, we use $\opt$, $\opt_{\tails}$ and $\alg$ to refer both to the value and the set itself in the optimal w.r.t. rewards, optimal w.r.t. samples, and the solution returned by our algorithm, respectively. 
Our analysis crucially relies on the following definition:

\begin{definition}[Saturation index]
For any set $S \in \L \cup \{E\}$, we denote $i_{\heads}(j,S) \in [2n]$ as the minimum index larger than $j$ in the greedy sample path such that there are exactly $r_S$ (the capacity of $S$) Y-values $W_{j'}$ such that $e_{j'}\in S$, $C_{j'}=\heads$, and $\freeT{j'}$ on the interval $\{j+1,\ldots,i_{\heads}(j,S)\}$.
Formally,
\begin{align*}
        i_{\heads}(j, S) = \arg\min\Bigg\{i > j \mid \sum_{\substack{j < j'  \leq i \text{ with } e_{j'}\in S \\ \text{Y-values}}}\1\{\C_{j'}= \heads \text{ and } \freeT{j'}\} = r_S \Bigg\},
\end{align*}
while we denote by $i_{\heads}(j,S)= \infty$ the fact that such an index does not exist. The definition of $i_{\tails}(j,S)$ follows analogously by replacing the condition $\C_{j'}= \heads$ with $\C_{j'}= \tails$.
\end{definition}

Similarly to the analysis in the previous sections, the following supporting event furnishes a link between the gambler's and the prophet's expected reward:

\begin{definition}[$\L$-supporting event]\label{def:laminarSupportingEvent}
For any $j\in [2n]$ with $e_j\in E_i$ for some $E_i \in \L$, we say that the \emph{supporting event} $\SS_{j}$
occurs if the following conditions hold simultaneously:
\begin{enumerate}
    \item $W_j$ is a $Y$-value satisfying $\freeT{j}$ and $\C_j= \heads$.
    \item If $\ell_1 = \min\{i_{\tails}(j, E_i) , i_{\heads}(j , E_i)\} < \infty$, then $\C_{\ell_1} = \tails$. In addition, if $\ell_2 = \min\{i_{\tails}(j, E) , i_{\heads}(j , E)\} < \infty$, then $\C_{\ell_2} = \tails$. 
    Simply put, the number of indices $j' > j$ in increasing order that are Y-values, $e_{j'} \in E_i$ (resp., $e_{j'} \in E$), and satisfy $C_{j'} = \tails$ reaches $r_{E_i}$ (resp., $r_{E}$) earlier than those satisfying $C_{j'} = \heads$.
\end{enumerate}
\end{definition}

Recall that our analysis must hold against an adversarial arrival order. In contrast to \cref{sec:matching,sec:transversal}, we are able to easily characterize the worst-case ordering for \cref{alg:laminar}:

\begin{restatable}{fact}{restateLaminarWorstCaseOrder} 
\label{claim:laminar:worstCaseOrder}
The worst-case arrival ordering for \autoref{alg:laminar} is in increasing order of rewards (even in the general case of laminar matroid).
\end{restatable}

The above fact allows us to reason about the probability that the algorithm collects \emph{exactly} a particular value. 
As a consequence, the supporting event $\SS_j$ in this case
is a sufficient condition for our policy to collect an element $e_j$ with value $W_j$ under this worst-case ordering.

\begin{lemma} \label{lem:survivingsufficient}
Assuming the worst-case arrival ordering (see \autoref{claim:laminar:worstCaseOrder}), for any index $j \in [2n]$ in the greedy path that satisfies $\SS_j$, the algorithm collects element $e_j$ with value $W_j$. Formally, for each $j \in [2n]$ it holds $\1\{e_j \in \alg \text{ and } C_j = \heads\} \geq \1\{\SS_j\}$.
\end{lemma}
\begin{proof}
Fix any index $j \in [2n]$ in the greedy path that satisfies $\SS_j$. By definition of $\SS_j$, it must be that $\C_j = \heads$, $\freeT{j}$ and $W_j$ is a Y-value. The combination of these facts immediately implies that by replacing the sample value of $e_j$ with $W_j$, we get $\opt(\tXv_{-e_j}, \X_{e_j}) > \opt(\tXv)$. 
Therefore, by definition of our algorithm, the element $e_j$ with reward $W_j$ would be collected in the online phase, if that does not violate independence. 
We now show that the event $\SS_j$ guarantees that $e_j$ will be feasible to collect when it arrives in the worst-case arrival ordering (that is, in increasing order of rewards).

Let $E_i \in \L$ and $E$ be the two sets which contain element $e_j$. Let us first focus on $E_i$. The assumed worst-case arrival order implies that the feasibility of collecting $e_j$ (with value $W_j$) can be compromised in $E_i$ only if the algorithm collects $r_{E_i}$ elements in $E_i$ of smaller value than $W_j$.
We claim that under event $\SS_j$, the above is not possible. Indeed, we first note that, in order for an index $j' > j$ in the greedy path to pose a threat for $j$ in $E_i$ (that is, for $e_{j'}$ to be accepted in place of $e_j$), $j'$ has to satisfy $\C_{j'} = \heads$, $\freeT{j'}$ and $e_{j'} \in E_i$. Further, any such $j'$ must be a Y-value, since otherwise it cannot satisfy $\opt(\tXv_{-e_{j'}}, \X_{e_{j'}}) > \opt(\tXv)$ -- thus it is automatically rejected by our algorithm. By definition of $\SS_j$, if $i_{\heads}(j, E_i) = \infty$, then set $E_i$ can never reach its capacity by collecting indices after $j$, thus implying that $j$ could be accepted in $E_i$ by the time $e_j$ arrives (in the worst-case arrival ordering). In the case where $i_{\heads}(j, E_i) < \infty$, by definition of $\SS_j$, it must hold that $i_{\tails}(j, E_i) < i_{\heads}(j, E_i)$. This means that there are no more than $r_{E_i} - 1$ elements $j < j' < i_{\tails}(j, E_i)$ that are Y-values satisfying $\C_{j'} = \heads$ and $e_{j'} \in E_i$. Thus, it suffices to show that the algorithm cannot collect any index $j' > i_{\tails}(j, E_i)$ in the online phase. 

By way of contradiction, assume that the algorithm collects an index $j' > i_{\tails}(j, E_i) $ in the online phase. Thus, it has to be that $\opt(\tXv_{-e_{j'}}, \X_{e_{j'}}) > \opt(\tXv)$, which in turn implies that $\freeT{j'}$ is true. However, this cannot be possible, since between $j$ and $j'$ in the greedy path, there exist at least $r_{E_i}$ indices $j < \ell < i_{\tails}(j, E_i) < j'$ such that $\C_{\ell} = \tails$ and $\freeT{\ell}$, thus, $\freeT{j'}$ cannot be true. 
Therefore, we can conclude that by the time $e_j$ arrives in the worst-case (increasing) ordering, the number of collected elements in $E_i$ is strictly smaller than $r_{E_i}$.

By repeating the above arguments with $E$ substituted for $E_i$, assuming $\SS_j$ is true, it follows that the number of collected elements before $e_j$ arrives in the worst-case ordering is strictly smaller than $r_E$.
Combining these two results yields the claimed inequality.
\end{proof}

The next step is to relate the probability of the supporting event $\SS_j$ for any Y-value $j \in [2n]$ to the probability that $W_j$ participates in the prophet's solution. As opposed to \cref{sec:matching,sec:transversal}, local coin flip manipulations are insufficient for mapping any configuration where $W_j$ is accepted by the prophet into a configuration satisfying $\SS_j$. Instead, we bound the probability of $\SS_j$ by studying the following two-player game: 

\begin{game}[Fair coins, nested bins and an unfair game] \label{game:twoplayer}
{
Consider two bins, R and B, with associated positive integers $r_R < r_B$. At each time step, a fair coin is tossed in one of the two bins. The bins are {\em nested} in the sense that every coin flip that falls into R also falls into B (the opposite, however, is \emph{not} true). The following game is played between two players, P1 (adversary) and P2: At each time step, and, crucially, \emph{before observing the outcome of the current coin}, P1 decides whether the coin is flipped into B or R (thus, contributing either to B or to \emph{both} bins). The game stops when both bins are {\em saturated} -- namely, when the number of heads or the number of tails in each bin reaches $r_R$ and $r_B$, respectively. P1 wins if {\em either} of the two bins is saturated with heads, while P2 wins if both bins are saturated with tails.
}
\end{game}

As we can see, the above game is biased towards P1 for two reasons: (i) he chooses the bin into which the next coin is flipped, and (ii) it suffices for any of the two bins to be saturated with heads in order to win. 
There is a simple characterization of the best strategy of P1 
(whose proof we defer to \autoref{sec:laminar-appendix}):

\begin{restatable}{lemma}{restateLaminarBiased}\label{claim:laminar:biased}
An optimal strategy for P1 in \autoref{game:twoplayer} is to first toss coins only into the blue bin until it gets saturated (with either heads or tails) and to then toss every coin into the red bin until the end of the game. In this case, the probability that P2 wins the game is exactly $\frac{1}{4}$.
\end{restatable}

In the next lemma, we lower bound the probability of the supporting event. The key insight is that, for any Y-value $j \in [2n]$, the conditional probability $\mathbb{P}[\SS_j \mid \C_j = \heads \text{ and }\freeT{j}]$ is lower bounded by the probability P2 wins in the above game. 

\begin{lemma} \label{lem:supporting:laminar:probability}
For each $j \in [2n]$ that is a Y-value, we have $\mathbb{P}[\SS_j] \geq \frac{1}{4} \mathbb{P}[\C_j = \heads \text{ and }\freeH{j}]$.
\end{lemma}

\begin{proof}
We fix any index $j \in [2n]$ in the greedy sample path such that $W_j$ is a Y-value. By the first condition in the definition of $\SS_j$, we have that: 
\begin{align*}
    \PRO{}{\SS_j} = \PROB{\SS_j \mid \C_j = \heads \text{ and }\freeT{j}} \cdot \PROB{C_j = \heads \text{ and }\freeH{j}},
\end{align*}
where the equality follows by \autoref{lem:symmetry}, since $j$ is a Y-value. 

In order to complete the proof, it suffices to show that 
$\PROB{\SS_j \mid \C_j = \heads \text{ and }\freeT{j}} \geq \frac{1}{4}$. Let us fix a {\em partial} configuration of the first $j$ coin flips $\{\C_{j'}\}_{j'=1}^{j}$ in the greedy sample path such that $\C_j = \heads$ and $\freeT{j}$. Recall that the above is possible since the event $\freeT{j}$ is completely determined by the coin flips $\{\C_{j'}\}_{j'=1}^{j-1}$. Note also, conditioning on such a partial configuration fixes the coin flips of some Z-values after $j$. Thus, the coin flips
of the Y-values corresponding to indices larger than $j$ are still independent and uniformly $\heads$ or $\tails$.

Let us now map the above setting to an instance of the two-player game described in \autoref{game:twoplayer}. 
We associate bins $R$ and $B$ with the set $E_i\in \L$ that contains $e_{j}$ and the ground set $E$, respectively, and set $r_R = r_{E_i}$ and $r_B = r_E$. 
We associate the time steps of the game with indices in
the greedy sample path $j'>j$ in decreasing order of $W_{j'}$ which are Y-values and satisfy $\freeT{j'}$. Notice that this condition can be checked simply by observing the realized coin flips up to $j'-1$.
At each time-step $t$ with corresponding index $j'$, we associate
the outcome of the coin flip in the game with $\C_{j'}$. In the case there are no more indices available in the greedy sample path, we simply append independent uniformly random coin flips to ensure that the game terminates. 
In the above setting, a \emph{possible} strategy at each time $t$ for P1 is to simply choose $R$ if the corresponding index $j'$ satisfies $e_{j'} \in E_i$, or to choose $B$ otherwise. In the case where time $t$ does not correspond to an index in the greedy sample path, P1 is allowed to choose any bin that is not already saturated. 

It is not hard to observe that if P2 wins the above game, this already implies that the supporting event $\SS_j$ is satisfied. Indeed, P2 wins if both R and B are saturated with $\tails$. In the case where both bins are saturated with $\tails$ from the greedy sample
path (that is, before starting to append independent tosses), then it has to be that $i_{\tails}(j,E) < i_{\heads}(j,E)$ and $i_{\tails}(j,E_i) < i_{\heads}(j,E_i)$. Otherwise, if the saturation occurs for any of the bins after appending independent coin tosses (that is, after all the indices of the greedy sample path have been exhausted), then it must be that either $i_{\tails}(j,E_i) = i_{\heads}(j,E_i) = \infty$ or $i_{\tails}(j,E) = i_{\heads}(j,E) = \infty$ (or both), in which case $\SS_j$ is trivially satisfied.

By \autoref{claim:laminar:biased}, the probability that P2 wins the above game (for \emph{any} adaptive policy of P1, and in particular, the one described above) is at least $\frac{1}{4}$. By the above argument, this implies that $\PROB{\SS_j \mid \C_j = \heads \text{ and }\freeT{j}} \geq \frac{1}{4}$, which completes the proof. 
\end{proof}

We remark that the proof of \autoref{lem:supporting:laminar:probability} reveals a fundamental difference between our approach and that of Ma et al. \cite{ma2016simulated}, where the analysis relies on the following trick: Instead of reasoning about the probability of an element being accepted in every layer \emph{simultaneously}, they reason about the probability of an element being \emph{rejected} from each layer \emph{independently}, and apply a union bound on these events. The fact that their secretary algorithm rejects a constant fraction of rewards is critical for making this union bound converge to a constant. On the contrary, our approach leverages the fact that the events that an element is accepted in each layer are {\em positively correlated}.

Using the above results, we obtain the following competitive guarantee.

\begin{theorem}\label{thm:competitveTwoLayerLaminar}
\textsc{single-sample laminar} is $8$-competitive
for two-layer laminar matroid.
\end{theorem}
\begin{proof}
By \autoref{lem:survivingsufficient}, for any index $j\in[2n]$ in the greedy sample path, if the supporting event $\SS_j$ is satisfied, then $W_j$ is accepted by the algorithm under the worst-case arrival ordering (see \autoref{claim:laminar:worstCaseOrder}). Hence, we can lower bound the expected reward collected by our algorithm as 
\begin{align*}
    \EX{}{\alg} \geq \sum_{\substack{j\in[2n] \\ \text{$Y$-values}}}
    W_j \cdot \PRO{}{\SS_j} 
    \geq \frac{1}{4} \sum_{\substack{j\in[2n] \\ \text{$Y$-values}}} W_j \cdot \PRO{}{\C_j = \heads \text{ and } \freeH{j}}
    \geq \frac{1}{8} \EX{}{\opt},
\end{align*}
where the second inequality follows by \autoref{lem:supporting:laminar:probability}. The last inequality follows by \autoref{lem:forgetZvals}, since the contribution of the Y-values is at least half of the prophet's expected reward.
\end{proof}

\section{From \texorpdfstring{$\alpha$}{a}-Partition to Single-Sample Prophet Inequalities}
\label{sec:reduction}

A matroid $\M=(E,\I)$ over a ground set $E$ is called a {\em simple} partition matroid\footnote{We refer to the partition matroid as ``simple'' in order to distinguish it from its common definition, where from each set $P_l$ of the partition, at most $p_l \geq 1$ elements can be selected.} if there exists some partition $\bigcup_{l \in [k]} P_l $ of $E$ such that $I \in \I$ if and only if $|I \cap P_l| \leq 1$ for each $l \in [k]$. Namely, $I$ contains at most one element from each set in the partition. Let $\opt(\M,\Xv)$ be the reward of the maximum independent set of a matroid $\M$ under a reward vector $\Xv$. We denote by $\Xv_A$ (resp., $\tXv_A$) the restriction of the reward (resp., sample) vector to the coordinates of a set $A \subseteq E$. We consider the following property of matroids, as defined\footnote{In fact, our definition is a slight adaptation of the analogous definition in \cite{BDGIT09}. Specifically, we additionally require that the transformation observes the values of a subset $S \subset E$ of the elements and that this subset is never included in ground set of the produced partition matroid.} in \cite{BDGIT09}:

\begin{definition}[Weak $\alpha$-Partition Property~\cite{BDGIT09}] \label{def:alphapartition}
Let $\M = (E , \I)$ be a matroid with element weights $\Xv \in \mathbb{R}^{|E|}_{\geq 0}$. We say that $\M$ satisfies an $\alpha$-partition property for some $\alpha \in [1, + \infty)$ if, after observing the weights of a (independent of $\Xv$ and possibly random) subset $S \subseteq E$ of the elements, one can define a simple partition matroid $\M' = (E' , \I')$ on a ground set $E' \subseteq E \setminus S$, such that:
\begin{enumerate}
    \item $\EX{\mathcal{R}}{\EX{\Xv \sim \DD}{\opt(\M', \Xv_{E'})}}  \geq \frac{1}{\alpha} \EX{\Xv \sim \DD}{\opt(\M, \Xv)}$,
    \item $\I' \subseteq \I$,
\end{enumerate}
where $\mathcal{R}$ is the possible randomness of the transformation (including the choice of $S$, i.e., the subset of queried elements). 
\end{definition}

Notice that, in addition to randomized transformations, the above definition also permits transformations that do not depend on the weights of \emph{any} samples (i.e., $S=\emptyset$), or that are deterministic.

\paragraph{Reduction.} 
Consider any matroid $\M = (E, \I)$ associated with a weight vector $\Xv \in \mathbb{R}^{|E|}_{\geq 0}$ that satisfies an $\alpha$-partition property for some $\alpha \geq 1$. We denote by $f(\M, \Xv)$ a black-box function that takes as an input the matroid $\M$ with an associated weight vector $\Xv$ and returns a matroid $\M' = (E', \I')$ that satisfies the properties of Definition \ref{def:alphapartition}. Let $S$ be the (potentially random or empty) subset of elements whose values are observed by the function in order to perform the transformation. 

We consider the following two-phased algorithm:

(1) In the offline phase, we first simulate the execution of the black-box function $f(\M, \tXv)$, which requests a subset of elements $S \subset E$, by feeding to the function the set of samples $\{\tX_e\}_{e\in S}$. Let $\M' = (E', \I')$ be the simple partition matroid returned by $f(\M, \tXv)$, with $E' \subseteq E \setminus S$. We denote by $\{P_l\}_{l \in [k]}$ the corresponding partition of $E'$, that is, $\bigcup_{l \in [k]}P_l = E'$ and $P_l \cap P_{l'} = \emptyset$ for all $l \neq l'$. For each group $l \in [k]$ of the partition, we define a threshold $T_l = \max_{e \in P_l} \{\tX_e\}$, which is the largest sample value over all the elements of the group. We initialize $I = \emptyset$ to be the set of collected elements.

(2) In the online phase of the algorithm, for each arriving element $e \in E$, we immediately reject the element (without even observing its value) if it does not belong to the ground set $E'$ of the partition matroid $\M'=(E', \I')$. In the case where $e \in E'$, let $l \in [k]$ be the group of the partition, i.e., the index such that $e \in P_l$. If $I \cap P_l = \emptyset$, i.e., no element of the group $P_l$ has been collected, and if $X_e > T_l$, i.e., the reward of the element is greater or equal to the threshold of the group, then the element is collected. Otherwise, the element is rejected.

As a consequence, we have the following meta-theorem, which gives rise to
the improved single-sample prophet inequalities shown in \autoref{table:main}:

\restateReduction*

\begin{proof}
We first establish the correctness of the policy, namely, the fact that the set $I$ of collected elements is an independent set of $\M = (E, \I)$. Let $\M'=(E',\I')$ be the simple partition matroid returned by $f$, as described in \autoref{def:alphapartition}. In the online phase of the algorithm, at most one element per group $l \in [k]$ is collected, which implies that $I \in \I'$. By construction of $\M'$ and by Definition~\ref{def:alphapartition}, we have that $\I' \subseteq \I$, which in turn implies that $I \in \I$. Further, note that, if the transformation of $\M$ into $\M'$ (and, thus, its simulation) can be performed efficiently, our policy runs in polynomial time.

We now prove the competitive guarantee of our policy. Let $\M'(S, \tXv_S)$ be the
matroid constructed as a function of the set $S$ (i.e., the elements queried by $f$) and the samples $\tXv_S$. Further, we denote $E'(S, \tXv_S)$ as the ground set of $\M'(S, \tXv_S)$. Note that, by \autoref{def:alphapartition}, the choice of $S$ is independent of the observed elements, as it is selected by $f(\M, \tXv)$ before any of their values is revealed. By \autoref{def:alphapartition} and taking expectation over all possible $\tXv \sim \DD$, we have:
\begin{align}
\EX{\mathcal{R}}{\EX{\Xv, \tXv \sim \DD}{ \opt(\M'(S, \tXv_S), \Xv_{E'(S, \tXv_S)})}}
&= \EX{\mathcal{R}}{\EX{\Xv, \tXv \sim \DD}{ \opt(\M'(S, \Xv_S), \Xv_{E'(S, \Xv_S)})}} \notag\\
&\geq \frac{1}{\alpha} \EX{\Xv, \tXv \sim \DD}{\opt(\M, \Xv)}. \label{eq:reduction:1}
\end{align}
In the above expression, the equality follows by replacing $\tXv_S$ with $\Xv_S$. 
Recall that, for a fixed set $S$, $\opt(\M'(S,\tXv_S),\Xv_{E'(S,\tXv_S)})$ depends only on $\tXv_S$ and $\Xv_{E'(S,\tXv_S)}$.
Since, by construction, $S\cap E'(S,\tXv_S)=\emptyset$ and $\tXv,\Xv$ are drawn independently
from the same product distribution $\DD$,
$(\tXv_S, \Xv_{E'(S,\tXv_S)})$ is identically distributed to $(\Xv_S, \Xv_{E'(S,\Xv_S)})$.
Thus, this step does not affect the expectation.
Further, the inequality follows by \autoref{def:alphapartition}.

We now focus on the expected reward collected by our policy. For any set $S$, we denote by $\{P_l(S, \tXv_S)\}_{l \in [k(S, \tXv_S)]}$ the $k(S, \tXv_S)$ groups of $\M'(S, \tXv_S)$, as a function of $S$ and the observed samples $\tXv_S$, with $\bigcup_{l \in k(S, \tXv_S)} P_l(S, \tXv_S) = E'(S, \tXv_S)$. For simplicity of exposition, in the rest of this proof, we suppress any dependence on $S$ and $\tXv_S$, and simply use $\M'$, $k$, $\{P_l\}_{l \in [k]}$ and $E'$. 

The online phase of our policy can be thought of as running $k$ single-choice \sspi{} instances in parallel, while the threshold $T_l$ takes the maximum value of any sample in $P_l$. Let $\alg_{l}$ for each $l \in [k]$ be the reward collected by our policy in group $P_l$ of the partition. Note that $\alg_{l}$ is a function of the rewards and the samples of the elements in $P_l$. By \cite{RWW20}, we know that setting the maximum sample value as a threshold and accepting the first element of value greater than that (if any), is a $2$-competitive policy for the single-choice problem. Thus, for any fixed $S\subset E$ and $\tXv_S$, we have:
\begin{align*}
\EX{\Xv_{E'}, \tXv_{E'} \sim \DD  }{\sum_{l \in [k]} \alg_l} 
\geq \frac{1}{2} \EX{\Xv_{E'} \sim \DD}{\sum_{l \in [k]} \max_{e \in P_l}\{\X_e\}}
= \frac{1}{2} \EX{\Xv_{E'} \sim \DD}{\opt(\M', \Xv_{E'})},
\end{align*}
where the last equality follows by the fact that the optimal solution of $\M'$ equals the sum of the maximum reward of each group in the partition.
By taking the expectation in the above expression over the vectors $\tXv_S$ and then over the randomness of the transformation (which includes the choice of $S$), we get:
\begin{align}
\EX{\substack{\mathcal{R}\\ \tXv_{S} \sim \DD}}{\EX{\Xv_{E'}, \tXv_{E'} \sim \DD}{\sum_{l \in [k]} \alg_l}} 
&\geq \frac{1}{2} \EX{\substack{\mathcal{R}\\ \tXv_S \sim \DD}}{\EX{\Xv_{E'} \sim \DD}{\opt(\M', \Xv_{E'})}}\nonumber\\
&= \frac{1}{2} \EX{\mathcal{R}}{ \EX{\Xv, \tXv \sim \DD}{ \opt(\M', \Xv_{E'})}}
\label{eq:reduction:2},
\end{align}
where the last equality follows by additionally taking the expectation over $\tXv_{E'}$, since for any fixed $S$, $\Xv_S$ and $\tXv_{E'}$ do not affect the value of $\opt(\M', \Xv_{E'})$.

Notice that the LHS of the above expression is exactly the expected reward collected by our algorithm. 
By combining inequalities \eqref{eq:reduction:1} and \eqref{eq:reduction:2}, we can conclude that our policy is $2\alpha$-competitive: 
\begin{align*}
\EX{\Xv, \tXv \sim \DD}{\alg(\M, \Xv)} = \EX{\substack{\mathcal{R}\\ \tXv_{S}\sim\DD}}{\EX{\Xv_{E'}, \tXv_{E'} \sim \DD}{\sum_{l \in [k]} \alg_l}} 
\geq \frac{1}{2 \alpha} \EX{\Xv, \tXv \sim \DD}{\opt(\M, \Xv)}. 
&\qedhere
\end{align*}
\end{proof}

{
\paragraph{Example: graphic matroids.}
We conclude this section by providing an application of \cref{thm:reduction} that yields a $4$-competitive \sspi{} for the case of {\em graphic matroids}. Recall that, given an undirected graph $\mathcal{G}(V,E)$, the family of independent sets of a graphic matroid $\M=(E,\I)$ (over the ground set of edges) consists of all acyclic subgraphs of $\mathcal{G}$. As proved in \cite{KP09}, the graphic matroid satisfies a $2$-partition property. 
This partitioning is constructed by
(i) choosing a uniformly random ordering $\sigma: V \to [|V|]$ on the nodes $V$ in $\mathcal{G}$ and (ii) defining for each $u \in V$ the subset of edges $E_u$ corresponding to neighbors of $u$ which come after $u$ in the ordering $\sigma$: $E_u = \{\{u,v\} \in E \mid \sigma(u) < \sigma(v)\}$. 
One may define a simple partition matroid over this partition $\{E_u\}_{u\in V}$ by adding the constraint that
at most one edge can be chosen from each $E_u$.
It can be verified that the above transformation satisfies \cref{def:alphapartition} with $\alpha = 2$. Thus, by applying \cref{thm:reduction} and running the rank-$1$ \sspi{} of \cite{RWW20} on each group $E_u$, the resulting policy is $4$-competitive.

\begin{remark}[A tight example for graphic matroids]\label{rem:tightExampleGraphic}
We remark that the analysis of the above policy for graphic matroids is tight. To see this, consider a {\em star graph} $\mathcal{G}(V, E)$,
where $V = \{s\} \cup \{s_1, \dots, s_k\}$ is the set of vertices and $E = \{\{s, s_i\}, \forall i \in [k]\}$ is the set of edges.
The sample and reward of each edge is drawn IID from $U\left[1-\frac{1}{k}, 1\right]$.
In this example, it can be verified that competitive ratio of our policy converges to $4$ as $k \to \infty$. 
Indeed, by construction of the $2$-partitioning, each edge $\{s,s_i\}$ belongs either to $E_s$ (which can contain up to $k$ edges) or to $E_{s_i}$ (which can contain at most one edge), each with probability half, 
independently of
all samples and rewards.
Given that $\{s,s_i\}$ belongs to $E_{s_i}$,
the policy collects the edge's reward with probability half (i.e., when the reward is larger than its corresponding sample). Additionally, although the partition $E_s$ contains $\frac{k}{2}$ edges in expectation, our policy allows \emph{at most} one of these edges to be collected.
The optimal policy, in contrast, collects the reward of \emph{every} edge with 
probability one.
\end{remark}
}

\section{Applications to Mechanism Design} \label{sec:mechanismdesign}
In this section, we describe how our results imply improved revenue and welfare guarantees for order-oblivious posted-price mechanisms with limited information in single and multi-dimensional settings. We point the reader to \cite{AKW18, dry10} for the necessary background on mechanism design and, specifically, on the design of prior-independent mechanisms.

We recall that a single-dimensional {\em sequential posted-price mechanism} (SPM) offers to sell a service to bidders (each arriving one at a time) at a take-it-or-leave-it price which depends on prior bids and the distributions of all bidders. An {\em order-oblivious posted-price mechanism} (OPM) is a SPM which maintains its competitive guarantee when the arrival of bidders is chosen adversarially. A mechanism is called \emph{truthful} if, regardless of the choices of other bidders, each individual bidder maximizes its utility by setting its bid to its private valuation.

Given a ``reasonable''\footnote{Specifically, a policy which is \emph{monotonic} with respect to the rewards -- increasing the value of any single reward can only increase the probability that item is selected.}
$\alpha$-competitive prophet inequality
for some downward-closed set system (which includes matching and matroids),
the framework designed in~\cite{AKW14,AKW18} 
produces truthful  single-dimensional order-oblivious posted price mechanisms (OPMs) with near-optimal revenue and welfare.
More specifically, they show that, in the case of (not necessarily identical) distributions satisfying the monotone hazard rate (MHR) condition, one can produce a $\alpha/2e$-revenue and a $\alpha/2$-welfare competitive mechanisms.
In addition, if the prophet inequality policy is comparison-based, one can produce an $\alpha/2$-revenue and welfare competitive mechanisms for the case of regular and identical distributions. To obtain this result, the authors use the concept of \emph{lazy sample reserves},
introduced by \cite{dry10}, in order to translate their \sspi{}s (which are all truthful, approximately welfare-optimal OPMs)
into mechanisms which are approximately revenue and welfare optimal \emph{simultaneously}.
This translation is achieved by drawing an independent sample from the distribution of each agent,
and accepting each winner (namely, an element that is chosen by the prophet inequality) if its valuation is greater than the lazy reserve.

The \sspi{}s of Azar et al.~\cite{AKW14}
obtained via the reduction to \oos{} have the property that,
for every reward collected by the algorithm, the corresponding
sample \emph{is not used by the algorithm}. Therefore,
these samples can be used as lazy reserves ``for free.''\footnote{Note, however, that their rehearsal algorithm for $k$-uniform matroids and their algorithm for degree-$d$ bipartite matching environments \emph{do not} satisfy this property -- as a result, the translation to mechanism design results requires an additional sample for the lazy reserves.}
In contrast, our improved competitive guarantees heavily rely on the utilization of all the samples. Even though all of our policies automatically reject any reward that is smaller than the corresponding sample, it is still unclear whether we can use these samples as lazy reserves (since they are now correlated with all of the prices). 
Thus, we simply require all of our mechanisms to have access to an additional sample
from each agent.
Resolving whether the samples given to the prophet can still be used as lazy reserves is an interesting open question which we leave as future work.

Given a \sspi{} $\mathcal{A}$ for some downward-closed system $\I$, the single-dimensional mechanism, due to Azar et al. \cite{AKW14,AKW18} can be described as follows:
\begin{enumerate}
    \item Run $\mathcal{A}$ using a single sample from $\DD$ to choose a set $W \in \I$ of winners with corresponding valuations $\{v_i \mid i\in W\}$ that approximately maximizes welfare.
    \item Obtain an additional reserve sample $\hat{r}\leftarrow \DD$ and accept any $i\in W$, only if $v_i\geq \hat{r}_i$.
\end{enumerate}

As a consequence of our improved \sspi{}s, we improve on most of the revenue and welfare competitive ratios of \cite{AKW14,AKW18} by a factor of at least $2$, at the cost of only a single additional sample. We summarize our results below:

\begin{corollary}[Our results + \cite{AKW14}]
Let $\I$ be a downward-closed set system, and let each $\DD_i$ be MHR. Then, there exists a truthful OPM that uses \emph{two} samples from $\DD$ and has the revenue and welfare guarantees given in \cref{table:mechMHR}. 
\end{corollary}

\begin{table}[H]%
\begin{center} 
 \begin{tabular}{||c c c c c||} 
 \hline
 Combinatorial set & Previous best & Reference & Our results & \\ 
  &  (welfare/revenue) &  & (welfare/revenue) & \\
  [0.5ex] 
 \hline\hline
 Bipartite matching & $512$ / $512e$ & \cite{AKW18} + \cite{FSZ18} & $64$ / $64e$  & Sec. \ref{sec:matching}  \\
 & $13.5$ / $13.5e$ & \cite{AKW14} & & \\
   &(degree $d$) & & &\\
   &($d^2+1$ samples) & & &\\
 General matching & - & - & $64$ / $64e$  & Sec. \ref{sec:matching} \\
 Transversal matroid & $32$ / $32e$ & \cite{AKW14} + \cite{DP08}  & $16$ / $16e$ & Sec. \ref{sec:transversal} \\ 
 Laminar matroid & $19.2$ / $19.2 e$ & \cite{AKW18} + \cite{ma2016simulated}  & $12\sqrt{3}$ / $12\sqrt{3}e$ & Thm. \ref{thm:reduction} + \cite{JSZ13}  \\ 
  & & & $16$ (2-layer) & Sec. \ref{sec:laminar} \\
 Graphic matroid & $16$ / $16e$ & \cite{AKW14} + \cite{KP09}  & $8$ / $8e$ & Thm. \ref{thm:reduction} + \cite{KP09}  \\ 
 Cographic matroid & $24$ / $24e$ & \cite{AKW14} + \cite{Soto11}  & $12$ / $12e$ & Thm. \ref{thm:reduction} + \cite{Soto11}  \\ 
 { Matroid of density $\gamma(\M)$} & $8\gamma(\M)$ / $8\gamma(\M)e$ & \cite{AKW14} + \cite{Soto11}  & $4\gamma(\M)$ / $4\gamma(\M)e$ & Thm. \ref{thm:reduction} + \cite{Soto11}  \\ 
     {Column $k$-sparse linear matroid} & {$8k$ / $8ke$} & \cite{AKW14} + \cite{Soto11}  & $4k$ / $4ke$ & Thm. \ref{thm:reduction} +  \cite{Soto11} \\
 [1ex] 
 \hline
\end{tabular}
\end{center}
\caption{Consequences for revenue and welfare-competitive OPMs, when $\I$ is a downward-closed set system and $\DD_i$ satisfy MHR. Unless otherwise indicated, all results for previous best use $1$ sample from the distribution, and all of our results use $2$ samples.}
	\label{table:mechMHR}
\end{table}%

We remark that our \sspi{}s also imply analogous improvements
over \cite{AKW14} in the case where each $\DD_i$ is identical and regular. We refer the interested reader to \cref{sec:mech-appendix} for details.

Finally, \cite{AKW14,AKW18} provide a framework for obtaining {\em multi-dimensional mechanisms} from single-dimensional OPMs for weighted bipartite matching under edge arrivals. They use the {\em copies environment} due to \cite{chawlaMultiParam} to reduce the design of approximately revenue-maximizing multi-dimensional mechanisms to that of revenue-maximizing \emph{single}-dimensional mechanisms. They combine this result with those of \cite{amdw13} and \cite{dry10} in the IID regular and the MHR case, respectively, to replace lazy monopoly reserves with lazy sample reserves. Our improved \sspi{} for the bipartite matching case allows us to provide the following result:

\begin{corollary}
For the multi-dimensional unit-demand mechanism design problem on (bipartite) matching environments, there exists a $64e$ revenue-competitive (resp., $64$ welfare-competitive) auction using two samples, in the case where agents' distributions of valuation satisfy MHR property (resp., are identical and regular).
\end{corollary}

\bibliographystyle{plainurl}
\bibliography{ref.bib}

\newpage
\appendix
\section{Transversal Matroid: Omitted Proofs}

\restateTransversalUniqueness*

\begin{proof}
For sake of contradiction, let us assume that there exists some configuration $\Cv$ and right vertex $r\in R$
for which $\SS_{j,r}$ is true for more than one index $j\in[2n]$.
Let $j_1, j_2 \in [2n]$ be the smallest such indices in the greedy sample path, with $j_1 < j_2$.
By definition of $\mathcal{S}_{j_1,r}$, we have that $j_1$ is a Y-value which satisfies $r = f_T(j_1)$ -- that is, $r$ is the smallest R-node in the fixed ordering such that $l_{j_1}$ would be matched with in $\opt'_T$, if $W_{j_1}$ was a sample value. Further, for $\ell > j_1$, as described in $\mathcal{S}_{j_1,r}$, namely, the smallest index after $j_1$ that is a Y-value, $\freeT{\ell}$ is true, and $r = f_T(\ell)$, we require that $\C_{\ell} = \tails$. Thus, since $\C_{j_1} = \heads$, it has to be that $r$ is matched to $l_{\ell}$ in $\opt'_T$, which in turn implies that, if $j_2 > \ell$, then $j_2$ cannot satisfy $\mathcal{S}_{j_2,r}$. 
Further, notice that the case $j_1 < j_2 < \ell$ is impossible, since $j_2$ (which is smaller than $\ell$) is a Y-value which satisfies
$\freeT{j_2}$ and $r=f_T(j_2)$, thus contradicting the minimality of $\ell$ as defined in $\SS_{j_1,r}$. Finally, we remark that if such an $\ell$ does not exist for $j_1$, then $j_2$ cannot exist, thus leading to a contradiction.
\end{proof}

\restateTransversalSufficiency*

\begin{proof}
Fix any configuration $\Cv$ and index $j\in[2n]$.
Recall that, by definition of the candidate right node (\autoref{def:candidateNodeTransversal}), there exists at most one right vertex $r\in R$ such that $f_T(j)=r$. Thus, for every fixed $j\in[2n]$,
at most one $\SS_{j,r}$ can be true.
Clearly, if there exists no $r \in R$ such that $\mathcal{S}_{j,r}$ holds, then the inequality follows trivially, since
the RHS is $0$.

Thus, let us assume otherwise, and take $r\in R$ to be the vertex such that the supporting event $\mathcal{S}_{j,r}$ is true. 
In that case, we wish to show that the reward of the L-node adjacent to $r$ in the solution of the policy is at least $W_j$, namely, $Q(r) \geq W_j$. 
By definition, the reward of $l_j$, i.e., the L-node corresponding to 
index $j$, is equal to $W_j$.
Consider the time where $l_j$ arrives in the online phase of the algorithm. As we show, there are two possible scenarios: either $l_j$ is matched to $r$ by our algorithm, or $r$ is already matched with another L-node of reward at least $W_j$.

We first claim that node $r$ is the smallest vertex in the fixed ordering of R-nodes such that $W_j \geq T_r$. Indeed, by construction of $\mathcal{S}_{j,r}$ it holds that $r = f_T(j)$. Note that this already implies that $W_j \geq T_r$. Since $j$ is a Y-value (by construction of the event $\SS_{j,r}$), it has to be that in $\opt'_T$, all the R-nodes that are adjacent to $l_j$ and smaller than $r$ in the fixed ordering are already matched with L-nodes of sample value greater than $W_j$ (since $\opt'_T$ is a ordered-maximal matching). Thus, their thresholds must exceed $W_j$. The above implies that if $l_j$ was presented first in the online phase of the algorithm, it would be matched to $r$, in which case $Q(r) \geq W_j$.

Our second claim is that if by the time $l_j$ arrives vertex $r$ is already matched, it has to be with an L-node of reward at least $W_j$. A sufficient event for this is that there exists no index $j' \in [2n]$ with $j < j' <\ell$ and $\ell$ as defined in $\mathcal{S}_{j,r}$, such that $j'$ is a Y-value, $\C_{j'} = \heads$, $\freeT{j'}$, and $r = f_T(j')$. 
Indeed, if $j'$ was a Z-value, then it is automatically rejected, by definition of our algorithm. In the same spirit, in the case where $\ell$ does not exist, it has to be that $W_j$ is the only reward whose corresponding element is adjacent to $r$ and can be accepted in the online phase, by definition of $\SS_{j,r}$. 
Assume, by way of contradiction, that such an index exists and let $j'$ be the smallest such index. In that case, by definition of $\mathcal{S}_{j,r}$ it has to be that $j' \equiv \ell$, which in turn implies that $\C_{j'} = \tails$, a contradiction. Therefore, if node $r$ is already matched by the time $l_{j}$ arrives, it has to be with a L-node of reward at least $W_j$, thus, $Q(r) \geq W_j$. 
\end{proof}

\restateTransversalProbability*

\begin{proof}
We first note that by \autoref{lem:symmetry}, since $j$ is a Y-value, it holds that $\PRO{}{\C_j = \heads \text{ and } \freeH{j}} = \PRO{}{\C_j = \heads \text{ and } \freeT{j}}$. Let us consider any configuration $\Cv \in \Cdist$ such that $\C_j = \heads$ and $\freeT{j}$. Clearly, for $\Cv$, the first property of $\mathcal{S}_{j,r}$ (for any possible $r \in R$) is satisfied. Since $\freeT{j}$ is true, there must exist an $r \in R$ such that $r = f_T(j)$. Let $\ell \in [2n]$ be the smallest index in the greedy path such that $\ell > j$, $\ell$ is a Y-value, satisfies $\freeT{\ell}$, and $r = f_T(\ell)$. Clearly, if it holds that $\C_{\ell} = \tails$ under configuration $\Cv$, then the event $\mathcal{S}_{j,r}$ is satisfied. In the case where $\C_\ell = \heads$, we can construct a configuration $\Cv'$ that is identical to $\Cv$, except for the coin associated to $\ell$ that is now $\C'_\ell = \tails$. Since $\ell$ is a Y-value, the first and second property in the definition of $\mathcal{S}_{j,r}$ remain satisfied under configuration $\Cv'$.

It can be easily verified that any configuration $\Cv'$ constructed in the above manner, satisfies $\mathcal{S}_{j,r}$. Further, note that any $\Cv'$ can be produced in the above way from a unique configuration $\Cv$. Using the fact that all feasible configurations are equiprobable, it immediately follows that $\PRO{}{\exists r \in R \text{ such that }\mathcal{S}_{j,r}} \geq \frac{1}{2} \PRO{}{\C_j = \heads \text{ and } \freeH{j}}$.
\end{proof}

\section{Truncated Partition Matroid: Omitted Proofs}
\label{sec:laminar-appendix}

\restateLaminarWorstCaseOrder*

\begin{proof}
Consider any arrival order $\sigma = \sigma_1,\sigma_2,\ldots,\sigma_n$ and suppose $\sigma$ is not in increasing order of rewards. Let $t\in [n]$ be the smallest time for which $\X_{\sigma_t} > 
\X_{\sigma_{t+1}}$ -- that is, the first time at which the arrival order
is not increasing. Since our policy is deterministic, we may assume w.l.o.g. that this happens at time $t=1$ (since until that time the current state of the collected elements is identical in both sequences).
We construct a new arrival sequence $\widehat{\sigma}$ such that
$\widehat{\sigma}_{1} = \sigma_{2}$, $\widehat{\sigma}_{2}=\sigma_{1}$, and
for every other $t>2$, $\widehat{\sigma}_{t}=\sigma_{t}$.
Then, it suffices to prove that the reward collected by our policy
under arrival sequence $\sigma$ is never
smaller than the reward collected under the arrival sequence $\widehat{\sigma}$. After establishing that, the proof follows easily by induction.

We remark that the only interesting case to consider is when, under $\sigma$, the policy accepts the larger element, $\X_{\sigma_1}$, and rejects the smaller element, 
$\X_{\sigma_2}$. Indeed, in every other case,
the performance of the policy remains unchanged,
since if the policy accepts only the smaller element under $\sigma$, 
it must do the same under $\widehat{\sigma}$. The same holds in the case where both elements are either accepted or rejected. Similarly,
we may assume w.l.o.g. that, under $\widehat{\sigma}$,
the policy accepts the smaller element, 
$\X_{\widehat{\sigma}_{1}}= \X_{\sigma_{2}}$ and rejects 
the larger element,
$\X_{\widehat{\sigma}_{2}} = \X_{\sigma_{1}}$.

Thus, we consider only the case when $X_{\sigma_{1}}$ is accepted under
arrival sequence $\sigma$ but not $\widehat{\sigma}$, and $\X_{\sigma_{2}}$
is accepted under $\widehat{\sigma}$ but not $\sigma$. Observe that, in this scenario, there must exist some set $L \in \mathcal{L}$ with capacity
$c_L=1$ such that both $\sigma_{1}$ and $\sigma_{2}$ are contained in $L$, 
since otherwise $\sigma_1$ and $\sigma_2$ could be added simultaneously.
Take $L$ to be the largest such set, and observe that this set is the same,
regardless of whether the arrival order is $\sigma$ or $\widehat{\sigma}$.
Consider any element $e\in L \setminus\{\sigma_1,\sigma_2\}$. Then, since $e$ arrives after time $t=2$, it must be rejected under either arrival
sequence. 
For every other set $L'\in\L$ such that $L'$ is not contained in $L$,
whether $\sigma_1$ or $\sigma_2$ was accepted by the algorithm, 
the capacity of each $L'$ is decreased by exactly $1$ after $t=2$.
Note that the policy makes the decision on whether to accept or reject
any arriving element based on an (order-oblivious) check for whether
the element would improve on the greedy solution with respect to samples,
and a feasibility check which, due to the laminar structure,
depends \emph{only} on the capacity of each set.
Therefore, the policy's decision to accept or reject an element $e \not\in L$
is the same for arrival sequences $\sigma$ and $\widehat{\sigma}$. This completes the proof, since the algorithm collects a smaller reward in the case of $\widehat{\sigma}$.
\end{proof}

\restateLaminarBiased*

\begin{proof}
Recall that P1 is unaware of the outcome of each coin flip before he decides on the bin where the coin is tossed. We consider the following strategy for P1: the coins are initially tossed in B until it becomes saturated (either with $\heads$, or with $\tails$). After this point, the rest of the coins are only tossed in R until it becomes saturated. Clearly, in the above strategy, the probability that P2 wins the game is exactly equal to $\frac{1}{4}$. This is because, due to the fairness of the coins, each bin becomes saturated independently either with $\heads$ or $\tails$ with equal probability, and both bins have to be saturated with $\tails$ in order for P2 to win the game. It suffices to show that this policy is the worst-case strategy for P2 or, equivalently, the optimal (adaptive) strategy for P1.

Let us introduce some necessary notation. We denote by $O_1, O_2, \ldots$ the stream of fair coin flips, where $O_t \in \{\heads, \tails\}$ for each $t \geq 1$. Let $Q_{i} \in \{\heads, \tails\}$ be the outcome of the coin flip during the $i$-th time P1 chooses R. We remark that index $i$ only counts the tosses in R and thus is generally different than $t$.
Let $\tau_1, \tau_2, \ldots$ be the points in time where the R bin is chosen by P1, depending on his strategy and the observed outcomes. Notice that the number of times where a coin flip can fall into R bin is at most $2r_R-1$.
Finally, we denote by $L(b, c)$ the event that bin $b \in \{B,R\}$ is saturated with $c \in \{\heads, \tails\}$. By overloading the notation, we also denote by $L(b, c, t)$ the event that this saturation occurs at time $t$.

Notice that, independently of the choices of P1, each coin flip necessarily falls into bin B and, thus, B is saturated with $\heads$ or $\tails$ with equal probability, that is, $\PRO{}{L(B,H)} = \frac{1}{2}$. Given any strategy, for the probability that P1 wins, we have:
\begin{align*}
    \PRO{}{\text{P1 wins}} &= \PRO{}{L(B,H)} + \PRO{}{L(B,T) \text{ and } L(R,H)} = \frac{1}{2} + \PRO{}{L(B,T) \text{ and } L(R,H)}.
\end{align*}
In order to show that the maximum probability of P1 winning the game is $\frac{3}{4}$ (thus, our proposed strategy is optimal), it suffices to show that for any strategy, it holds $\PRO{}{L(B,T) \text{ and } L(R,H)} \leq \frac{1}{4}$.

Let $\QH$ be the family of sequences of $2r_R -1$ coin flips (that is, the maximum possible) in R that lead to the saturation of R bin with $\heads$. In other words, $\QH$ is the family of all sequences of coin flips in R where the number of $\heads$ reaches $r_R$ before the number of $\tails$. We denote by $\QHtill{i}$ the set of subsequences $\Qtill{i}=(Q_1,\ldots,Q_i)$ corresponding to some $\Q=(Q_1,\ldots,Q_{2r_R-1})\in\QH$. Thus, we may write:
\begin{align*}
\PRO{}{L(B,T) \text{ and } L(R,H)} 
&= \sum_{t_B \geq 1} \PRO{}{{L(R,H)}\mid L(B,T,t_B)} \PRO{}{L(B,T,t_B)} \\
&= \sum_{t_B \geq 1} \PRO{}{\Q \in \QH \mid L(B,T,t_B)} \PRO{}{L(B,T,t_B)},
\end{align*}
where $t_B$ is the saturation time of bin $B$.

Let us fix any $t_B$. In the rest of this proof, for ease of notation, we assume that all probabilities are conditional on $L(B,T,t_B)$, and use the shorthand $\PROtB{\cdot} = \PROB{\cdot \mid L(B,T,t_B)}$. Recall that we define by $\tau_i$ the point in time where a coin is tossed in R for the $i$-th time. In order for $\tau_i$ to be well-defined for any $i \in [2 r_R - 1]$, we can assume w.l.o.g. that P1 always chooses bin R exactly $2r_R -1$ times, since the tosses in R after it saturates do not matter. We have:
\begin{align}
\PROtB{\Q\in\QH} 
&= \sum_{0 < t_1 < t_2 < \ldots < t_{2 r_R -1}} \PROtB{\tau_i = t_i,~\forall i \in [2 r_R -1] \text{ and } (O_{t_1},\ldots,O_{t_{2r_R-1}})\in\QH} \notag \\ 
&= \sum_{0 < t_1 < t_2 < \ldots < t_{2 r_R -1}} \PROtB{(O_{t_1},\ldots,O_{t_{2r_R-1}})\in\QH} \notag
\\& \qquad \qquad \qquad \qquad \qquad \cdot \PROtB{\tau_i = t_i,~\forall i \in [2 r_R -1] \mid (O_{t_1},\ldots,O_{t_{2r_R-1}})\in\QH}. \label{eq:biased:0}
\end{align}

Now, we may decompose the above conditional probability as follows:
\begin{align}
    &\PROtB{\tau_i = t_i ~ \forall i\in[2 r_R -1] \mid (O_{t_1},\ldots,O_{t_{2r_R-1}})\in\QH} \notag \\
    &=
    \prod_{i=1}^{2r_R-1} \PROtB{\tau_i = t_i \mid (O_{t_1},\ldots,O_{t_{2r_R-1}})\in\QH, \tau_s = t_s ~ \forall s\in [i-1]} \notag \\
    &=
    \prod_{i=1}^{2r_R-1} \PROtB{\tau_i = t_i \mid (O_{t_1},\ldots,O_{t_{i-1}})\in\QHtill{i-1} \text{ and } \tau_s = t_s ~ \forall s\in [i-1]} \label{eq:biased:1},
\end{align}
where the second equality follows by the fact that the decision of P1 on tossing the coin at time $t$ in either R or B depends only on the outcomes and decisions \emph{before} $t$.

Let us now focus on the probability $\PROtB{(O_{t_1},\ldots,O_{t_{2r_R-1}})\in\QH}$. Since we condition on the event $L(B,T,t_B)$, we know that, by time $t_b$, the total number of $\tails$ is greater than the total number of $\heads$, and additionally $O_{t_B} = \tails$, by definition of $t_B$. Therefore, for any sequence of outcomes that lead to the saturation of R with $\heads$, it is the case that:
\begin{align}
\PROtB{(O_{t_1},\ldots,O_{t_{2r_R-1}})\in\QH} \leq \PROtB{(O_{t_B+1},\ldots,O_{t_B+2r_R-1})\in\QH} = \frac{1}{2}. \label{inq:biased:2}
\end{align}

The above inequality holds since, after time $t_B$, all the coin flips are independent (even conditioned on $L(B,\tails,t_B)$), while before $t_B$, their joint distribution is a uniformly random permutation of $r_B-1$ $\tails$ and $t_B - r_R \leq r_B-1$ $\heads$ (note that conditioning on $L(B,\tails,t_B)$ fixes the number, but not the positions of $\tails$ and $\heads$ on the interval $[t_B - 1]$).

Let us define: $p_{i,t_i} = \PROtB{\tau_i = t_i \mid (O_{t_1},\ldots,O_{t_{i-1}})\in\QHtill{i-1} \text{ and } \tau_s=t_s ~ \forall s\in[i-1]}$.
By the above definition and using \eqref{eq:biased:0}, \eqref{eq:biased:1}, and \eqref{inq:biased:2}, we have that:

\begin{align*}
\PROtB{\Q\in\QH}
&\leq \frac{1}{2}
\sum_{0 < t_1 < t_2 < \ldots < t_{2 r_R -1}} \PROtB{\tau_i = t_i,~\forall i \in [2 r_R -1] \mid (O_{t_1},\ldots,O_{t_{2r_R-1}})\in\QH}\\ 
&= \frac{1}{2}\sum_{0 < t_1 < t_2 < \ldots < t_{2 r_R -1}}
\prod_{i=1}^{2r_R-1} p_{i,t_i}\\
&= \frac{1}{2}\sum_{0 < t_1 < t_2 < \ldots < t_{2 r_R - 2}}
\prod_{i=1}^{2r_R-2}
p_{i,t_i}  \sum_{t_{2r_R-1} > t_{2r_R-2}} p_{2r_R-1, t_{2r_R-1}} \\
&= \frac{1}{2}\sum_{0 < t_1 < t_2 < \ldots < t_{2 r_R - 2}}
\prod_{i=1}^{2r_R-2}
p_{i,t_i},
\end{align*}

where the second equality follows since $p_{i,t_i}$ depends
only on $\{t_1,\ldots,t_i\}$, 
and the last equality follows since $\tau_{2r_{R}-1}$ is always
some time larger than $t_{2r_R-2}$, conditioned on 
$\tau_{2r_R-2}=t_{2r_R-2}$.

By repeating the above argument for each $p_{i,t_i}$, we conclude that $ \PROtB{\Q\in\QH} \leq \frac{1}{2}$. By combining the above facts, we arrive at the desired inequality:
\begin{align*}
\PRO{}{L(B,T) \text{ and } L(R,H)} &= \sum_{\substack{t_B \geq 1}} \PRO{}{\mathbf{Q} \in \mathcal{Q}_{\heads}~|~L(B,T,t_B)} \PRO{}{L(B,T,t_B)} \leq \frac{1}{2} \sum_{t_B \geq 1}\PRO{}{L(B,T,t_B)} = \frac{1}{4}, 
\end{align*}
where the last equality follows by the fact that, independently of the actions of P1, B is saturated with $\tails$ with probability exactly half.
\end{proof}

\section{Mechanism Design: Omitted Results}
\label{sec:mech-appendix}

Using the mechanism described in \cref{sec:mechanismdesign},
we note that our \sspi{}s also imply improved truthful, single-dimensional OPMs with improved revenue and welfare guarantees compared to \cite{AKW14} at the cost of only a \emph{single} additional sample in the case of IID regular distributions. We summarize our results below:

\begin{corollary}[Our results + \cite{AKW14}]
Let $\I$ be a downward-closed set system, and let each $\DD_i$ be identical and regular.
Then, there exists a truthful OPM that uses \emph{two} samples from each $\DD_i$ and
has the welfare/revenue guarantees given in \cref{table:mechIdentical}.
\end{corollary}

\begin{table}[H]%
\begin{center} 
 \begin{tabular}{||c c c c c||} 
 \hline
 Combinatorial set & Previous best & Reference & Our results & \\ 
  &  (welfare/revenue) & & (welfare/revenue) & \\
  [0.5ex] 
 \hline\hline
 Bipartite matching & $512$ & \cite{AKW18} + \cite{FSZ18} & 64 & Sec. \ref{sec:matching}  \\
  & $27$ ($d$-degree) & \cite{AKW14} & & \\
   & ($d^2+1$ samples)& & & \\
 General matching & - & - & 64 & Sec. \ref{sec:matching} \\
 Transversal matroid & $32$ & \cite{AKW14} + \cite{DP08}  & $16$ & Sec. \ref{sec:transversal}  \\ 
 Laminar matroid & $19.2$ & \cite{AKW18} + \cite{ma2016simulated} &  $12\sqrt{3}$ & Thm. \ref{thm:reduction}  + \cite{JSZ13}  \\ 
   & & & $16$ (2-layer) & Sec. \ref{sec:laminar}\\
 Graphic matroid & $16$ & \cite{AKW14} + \cite{KP09}  & $8$ & Thm. \ref{thm:reduction}  + \cite{KP09}  \\ 
 Cographic matroid & $24$ & \cite{AKW14} + \cite{Soto11}  & $12$ & Thm. \ref{thm:reduction}  + \cite{Soto11}  \\ 
 { Matroid of density $\gamma(\M)$} & $8\gamma(\M)$ & \cite{AKW14} + \cite{Soto11}  & $4\gamma(\M)$ & Thm. \ref{thm:reduction} + \cite{Soto11}  \\ 
     {Column $k$-sparse linear matroid} & {$8k$} & \cite{AKW14} + \cite{Soto11}  & $4k$ & Thm. \ref{thm:reduction} + \cite{Soto11} \\
 [1ex] 
 \hline
\end{tabular}
\end{center}
\caption{Consequences for revenue and welfare-competitive OPMs, when $\I$ is a downward-closed set system and $\DD_i$ are IID regular. Unless otherwise indicated, all results for previous best use $1$ sample from the distribution,
	and all of our results use $2$ samples.}
\label{table:mechIdentical}
\end{table}%

\end{document}